\documentclass{llncs}
\usepackage{silence}
\usepackage[T1]{fontenc}
\usepackage{anyfontsize}
\usepackage[table]{xcolor}
\usepackage{pifont}

\usepackage{amsmath}
\usepackage{xfrac}
\usepackage{mathtools}
\usepackage{nicefrac}
\ActivateWarningFilters[hcounter]
\usepackage[adversary,ff,keys,operators,primitives,probability,sets,oracles,lambda,advantage,notions,asymptotics,logic]{cryptocode}
\DeactivateWarningFilters[hcounter]

\usepackage{paralist}
\usepackage[abbreviations]{glossaries-extra}
\usepackage[hidelinks]{hyperref}
\usepackage[capitalise,nameinlink]{cleveref}

\usepackage{makecell}
\usepackage{booktabs}
\usepackage[para]{threeparttable}

\usepackage{listings}

\usepackage{xspace}

\glsdisablehyper

\newacronym{aead}{AEAD}{Authenticated Encryption with Associated Data}
\newacronym{iv}{IV}{Initialization Vector}

\newabbreviation{co}{CO}{ciphertext-only}
\newabbreviation{kpa}{KPA}{Known Plaintext Attack}
\newabbreviation{cpa}{CPA}{Chosen Plaintext Attack}
\newabbreviation{ppt}{PPT}{probabilistic polynomial-time}

\newabbreviation{mte}{MtE}{MAC-then-Encrypt}
\newabbreviation{eam}{EaM}{Encrypt-and-MAC}
\newabbreviation{etm}{EtM}{Encrypt-then-MAC}

\newacronym{ietf}{IETF}{Internet Engineering Task Force}
\newacronym{rfc}{RFC}{Request for Comments}
\newacronym{ssh}{SSH}{Secure Shell}
\newabbreviation{bpp}{BPP}{Binary Packet Protocol}
\newacronym{tls}{TLS}{Transport Layer Security}
\newacronym{tcp}{TCP}{Transmission Control Protocol}
\newacronym{dtls}{DTLS}{Datagram Transport Layer Security}
\newacronym{ipsec}{IPsec}{Internet Protocol Security}
\newacronym{srtp}{SRTP}{Secure Real-time Transport Protocol}
\newacronym{esp}{ESP}{Encapsulating Security Payload}
\newacronym{ah}{AH}{Authentication Header}
\newacronym{ascii}{ASCII}{American Standard Code for Information Interchange}

\newacronym{cbc}{CBC}{Cipher Block Chaining}
\newacronym{ctr}{CTR}{Counter Mode}
\newacronym{cfb}{CFB}{Cipher Feedback Mode}
\newacronym{ccm}{CCM}{Counter with CBC-MAC}
\newacronym{gcm}{GCM}{Galois/Counter Mode}
\newacronym{ofb}{OFB}{Output Feedback Mode}

\newacronym{mac}{MAC}{Message Authentication Code}
\newabbreviation{prp}{PRP}{Pseudorandom Permutation}
\newabbreviation{prg}{PRG}{Pseudorandom Generator}

\newacronym{aes}{AES}{Advanced Encryption Standard}
\newacronym{aesgcm}{AES-GCM}{Advanced Encryption Standard in Galois/Counter Mode}
\newacronym{rc4}{RC4}{Rivest Cipher 4}
\newabbreviation{dhke}{DHKE}{Diffie-Hellman Key Exchange}

\newacronym{rq}{RQ}{Research Question}

\newacronym{beast}{BEAST}{Browser Exploit Against SSL/TLS}

\newabbreviation{mitm}{MitM}{Man-in-the-Middle}

\newacronym{acce}{ACCE}{Authenticated and Confidential Channel Establishment}
\newacronym{ake}{AKE}{Authenticated Key Exchange}
\newcommand{\modesModeled}{eight\xspace}

\newcommand{\X}{\pcnotionstyle{X}\xspace}
\newcommand{\CO}{\pcnotionstyle{CO}\xspace}
\newcommand{\KPA}{\pcnotionstyle{KPA}\xspace}
\newcommand{\CPA}{\pcnotionstyle{CPA}\xspace}
\newcommand{\intpcsa}{\pcnotionstyle{CINT\pcmathhyphen{}PCSA}\xspace}
\newcommand{\intpcsax}{\pcnotionstyle{\intpcsa\pcmathhyphen\X}\xspace}
\newcommand{\intpcsaco}{\pcnotionstyle{\intpcsa\pcmathhyphen\CO}\xspace}
\newcommand{\intpcsakpa}{\pcnotionstyle{\intpcsa\pcmathhyphen\KPA}\xspace}
\newcommand{\intpcsacpa}{\pcnotionstyle{\intpcsa\pcmathhyphen\CPA}\xspace}
\newcommand{\indpcsaccaii}{\pcnotionstyle{CIND\pcmathhyphen{}PCSA\pcmathhyphen{}CCA2}\xspace}
\newcommand{\intctxt}{\pcnotionstyle{INT\pcmathhyphen{}CTXT}\xspace}
\newcommand{\intptxt}{\pcnotionstyle{INT\pcmathhyphen{}PTXT}\xspace}
\newcommand{\sufcma}{\pcnotionstyle{SUF\pcmathhyphen{}CMA}\xspace}
\newcommand{\sprp}{\pcnotionstyle{SPRP}}

\renewcommand{\st}[1][]{\ensuremath{\mathsf{st}\if\relax\detokenize{#1}\relax\else_{#1}\fi}\xspace}
\newcommand{\sqn}[1][]{\ensuremath{\mathsf{sqn}\if\relax\detokenize{#1}\relax\else_{#1}\fi}\xspace}
\newcommand{\IV}[1][]{\ensuremath{\mathsf{IV}\if\relax\detokenize{#1}\relax\else_{#1}\fi}\xspace}
\newcommand{\pos}[1][]{\ensuremath{\mathsf{pos}\if\relax\detokenize{#1}\relax\else_{#1}\fi}\xspace}
\newcommand{\ctr}[1][]{\ensuremath{\mathsf{ctr}\if\relax\detokenize{#1}\relax\else_{#1}\fi}\xspace}
\newcommand{\ictr}[1][]{\ensuremath{\mathsf{ictr}\if\relax\detokenize{#1}\relax\else_{#1}\fi}\xspace}

\newcommand{\advst}[1][]{\ensuremath{\xi\if\relax\detokenize{#1}\relax\else_{#1}\fi}\xspace}
\newcommand{\intst}[1][]{\ensuremath{\eta\if\relax\detokenize{#1}\relax\else_{#1}\fi}\xspace}

\newcommand{\res}{\ensuremath{\mathit{res}}\xspace}
\newcommand{\idx}[1][]{\ensuremath{j\if\relax\detokenize{#1}\relax\else_{#1}\fi}\xspace}
\newcommand{\msgStore}[1][]{\ensuremath{M\if\relax\detokenize{#1}\relax\else{[#1]}\fi}\xspace}
\newcommand{\challbit}{\ensuremath{b}\xspace}
\newcommand{\challbitp}{\ensuremath{b'}\xspace}
\newcommand{\setflag}[1][]{\ensuremath{p\if\relax\detokenize{#1}\relax\else_{#1}\fi}\xspace}
\newcommand{\rcvfailflag}{\ensuremath{f_r}\xspace}

\newcommand{\kenc}{\ensuremath{\key_e}\xspace}
\newcommand{\kauth}{\ensuremath{\key_a}\xspace}
\newcommand{\salt}{\ensuremath{\mathsf{salt}}\xspace}

\newcommand{\SND}{\oracle[SND]\xspace}
\newcommand{\SNDCO}{\ensuremath{\SND\pcmathhyphen{}\oracle[CO]}\xspace}
\newcommand{\SNDKPA}{\ensuremath{\SND\pcmathhyphen{}\oracle[KPA]}\xspace}
\newcommand{\SNDCPA}{\ensuremath{\SND\pcmathhyphen{}\oracle[CPA]}\xspace}
\newcommand{\SNDX}{\ensuremath{\SND\pcmathhyphen{}\oracle[X]}\xspace}
\newcommand{\SET}{\oracle[SET]\xspace}
\newcommand{\RCV}{\oracle[RCV]\xspace}
\newcommand{\Exp}{\oracle[Exp]\xspace}

\newcommand{\cbcmode}{\ensuremath{\mathsf{CBC}}\xspace}
\newcommand{\cbcenc}{\ensuremath{\cbcmode.\enc}\xspace}

\newcommand{\ctrmode}{\ensuremath{\mathsf{CTR}}\xspace}
\newcommand{\ctrenc}{\ensuremath{\ctrmode.\enc}\xspace}

\newcommand{\gcmmode}{\ensuremath{\mathsf{GCM}}\xspace}
\newcommand{\gcmenc}{\ensuremath{\gcmmode.\enc}\xspace}

\newcommand{\streammode}{\ensuremath{\mathsf{Stream}}\xspace}
\newcommand{\streamenc}{\ensuremath{\streammode.\enc}\xspace}

\newcommand{\etm}{\ensuremath{\mathsf{\glsxtrshort{etm}}}\xspace}
\newcommand{\eam}{\ensuremath{\mathsf{\glsxtrshort{eam}}}\xspace}
\newcommand{\eamcbc}{\ensuremath{\mathsf{\eam\pcmathhyphen{}\cbcmode}}\xspace}
\newcommand{\etmcbc}{\ensuremath{\mathsf{\etm\pcmathhyphen{}\cbcmode}}\xspace}
\newcommand{\eamctr}{\ensuremath{\mathsf{\eam\pcmathhyphen{}\ctrmode}}\xspace}
\newcommand{\etmctr}{\ensuremath{\mathsf{\etm\pcmathhyphen{}\ctrmode}}\xspace}
\newcommand{\eamstr}{\ensuremath{\mathsf{\eam\pcmathhyphen{}\streammode}}\xspace}
\newcommand{\etmstr}{\ensuremath{\mathsf{\etm\pcmathhyphen{}\streammode}}\xspace}

\newcommand{\msg}{\ensuremath{m}\xspace}
\newcommand{\cmsg}{\ensuremath{c}\xspace}
\newcommand{\aeadNonce}{\ensuremath{N}\xspace}
\newcommand{\aeadAD}{\ensuremath{A}\xspace}
\newcommand{\aeadTag}{\ensuremath{\tau}\xspace}
\newcommand{\sfHeader}{\ensuremath{\mathrm{H}}\xspace}
\newcommand{\aeadAlg}{\pcalgostyle{AEAD}\xspace}

\newcommand{\aeadenc}{\pcalgostyle{\aeadAlg.\enc}\xspace}

\newcommand{\sfAlg}{\pcalgostyle{SF}\xspace}
\newcommand{\init}{\pcalgostyle{Init}\xspace}
\newcommand{\sfsetup}{\pcalgostyle{\sfAlg.\init}\xspace}
\newcommand{\sfenc}{\pcalgostyle{\sfAlg.\enc}\xspace}
\newcommand{\sfdec}{\pcalgostyle{\sfAlg.\dec}\xspace}
\newcommand{\ssh}{\pcalgostyle{SSH}\xspace}
\newcommand{\padssh}[1]{\ensuremath{\mathsf{Pad}_\ssh(#1)}\xspace}
\newcommand{\padssha}[1]{\ensuremath{\mathsf{Pad}^\ast_\ssh(#1)}\xspace}

\newcommand{\lastBlock}[1]{\ensuremath{\mathsf{LastBlock}(#1)}\xspace}

\newcommand{\lineref}[1]{\hyperref[#1]{line~\ref*{#1}}}
\newcommand{\blklen}{\ensuremath{n_\enc}\xspace}
\newcommand{\blkAlign}{\ensuremath{B}\xspace}
\newcommand{\byte}[1]{\ensuremath{\texttt{0x#1}}\xspace}
\newcommand{\len}[1]{\ensuremath{\mathsf{len}(#1)}\xspace}

\newcommand{\yes}{\ding{51}\xspace}

\newcommand{\no}{\ding{56}\xspace}

\begin{document}

\title{Formally Modeling the Terrapin Attack on SSH}

\author{Jörg Schwenk\inst{1}\orcidID{0000-0001-9315-7354} \and
Fabian Bäumer\inst{1}\orcidID{0009-0006-5569-6625} \and
Marcus Brinkmann\inst{1}\orcidID{0000-0001-5649-6357}}
\authorrunning{J. Schwenk et al.}
\institute{Ruhr University Bochum, Bochum, Germany}

\maketitle

\begin{abstract}
  The Terrapin attack against \gls{ssh} channel integrity (USENIX Security 2024) used a novel attack vector: attacks on the \emph{channel state}. Surprisingly, not all \glsxtrshort{aead} modes of \gls{ssh} were equally affected by this attack, and it remained an open question if ``unaffected'' meant ``secure''.

Existing formal models for secure channels are based on stateful encryption. However, these models do not define what the channel state is and how it is used as input to the different \gls{aead} modes.

In this paper, we propose a formal model for channel integrity under partially chosen state. Applied to the Terrapin attack, the chosen state is the SSH sequence number. It uses an abstract stateful encryption interface, for which we provide pseudocode descriptions for the \modesModeled most prominent \gls{aead} modes used in \gls{ssh}. By varying the $\SND$ oracle, we can model ciphertext-only (\glsxtrshort{co}; the Terrapin attack)\glsunset{co}, known-plaintext (\glsxtrshort{kpa})\glsunset{kpa}, and chosen-plaintext (\glsxtrshort{cpa})\glsunset{cpa} attacks. This allows us to establish concrete bounds on the security of the \gls{aead} modes. We find that all three \gls{etm} modes and ChaCha20-Poly1305 in \gls{ssh} are insecure in the \gls{co} model. AES-GCM is the only cipher secure in all three model variants. Going beyond Terrapin, we show that \gls{eam} with a \gls{cbc} cipher is secure, even in the \gls{kpa} model. In particular, we describe a novel BEAST-like chosen-plaintext attack on the channel integrity of \gls{eam}-\gls{cbc}, which separates the \gls{kpa} and \gls{cpa} models for this scheme.

\end{abstract}

\RenewCommandCopy{\state}{\cryptocodestate}

\section{Introduction}
\label{sec:intro}
\subsection{Motivation}

\subsubsection{Terrapin Attack.}

In 2024, an attack on the \gls{ssh} protocol was published~\cite{USENIX:BauBriSch24}, which attracted interest from both academia and industry. One insight from that paper is that the stateful encryption schemes deployed in \gls{ssh} are affected differently by the attack. While the channel integrity for ChaCha20-Poly1305 and \gls{etm} is broken, \glsxtrshort{gcm} and \gls{eam} modes are unaffected. This is somehow counterintuitive: if the main channel state can be controlled by an adversary, why is there this difference? And what does ``unaffected'' mean: does it mean ``secure'' in any formal sense? While the informal treatment in~\cite{USENIX:BauBriSch24} has the advantage that it can also include the application layer, a rigorous answer to these questions requires a formal analysis of the underlying constructions. Thus we raise the question:

\begin{quote}
    \emph{\Gls{rq}~1: Can we extend existing security models for secure channels to include Terrapin-like attacks on channel state and use these models to explain the differences in the security of SSH cipher modes?}
\end{quote}

\subsubsection{\glsfmtshort{ssh}.}

\Gls{ssh}~2.0, similar to \gls{tls}, establishes an authenticated confidential channel between a client and a server. The secure channel uses the \gls{bpp} and stateful encryption schemes for the packets. The \Gls{bpp}, key exchange, and server authentication are described in~\cite{rfc4253}, while two other documents specify client authentication~\cite{rfc4252} and the application-layer protocol~\cite{rfc4254}. Cipher modes in \gls{ssh} may differ significantly from other protocols; a striking example is ChaCha20-Poly1305, where the \gls{ssh} version (\cref{fig:ccp}) structurally differs from \gls{rfc}~8439~\cite{rfc8439}; thus we describe each mode in detail in \cref{sec:modes}.

The Terrapin attack~\cite{USENIX:BauBriSch24} is enabled by two design flaws in \gls{ssh}: (1) The \gls{ssh} handshake allows for many optional messages, which are not included in the hash of the partial handshake transcript; (2) SSH sequence numbers are not reset when key material changes. In the attack, the attacker manipulates the sequence numbers by injecting optional messages into the handshake and resynchronizes the sequence numbers by removing an equal number of messages at the beginning of the secure channel (see \cref{sec:bg:terrapin}). In this paper, we model the handshake message injection vector as an oracle available to the attacker and formally prove the consequences for \modesModeled different \gls{ssh} stateful encryption schemes.

\subsubsection{Secure Channels, \glsfmtshort{aead} and Stateful Encryption.}

Secure channels, like the \gls{tls} Record Layer and the \gls{ssh} \gls{bpp}, are constructed from \emph{\gls{aead} cipher modes} by using channel-specific \emph{state information}---typically sequence numbers. In fact, Rogaway~\cite{CCS:Rogaway02} identified secure network channels as the origin of \gls{aead}. Theoretical constructions of secure channels have been described, among others, by Boyd~et~al.~\cite{RSA:BHMS16} and Fischlin~et~al.~\cite{C:FGMP15}. However, the cryptographic primitive used in these constructions is not \gls{aead} but \emph{stateful encryption}, a notion introduced by Bellare, Kohno, and Namprempre~\cite{CCS:BelKohNam02} and refined in~\cite{C:JKSS12,C:KraPatWee13,RSA:BHMS16,C:FGMP15}.

Unfortunately, the interfaces of these two primitives are different; see \cref{tab:syntax}. \Gls{aead} input consists of a secret key~\key, a unique nonce~\aeadNonce, associated data~\aeadAD, and plaintext~\msg; the output is a ciphertext~\cmsg and an authentication tag~\aeadTag. In stateful encryption, a key~\key, a channel state~\st, and a message~\msg are the inputs; the output is an authenticated header~\sfHeader, the ciphertext~\cmsg, and an updated state~$\st'$. These differences make it challenging to compare formal models for secure channels with their real-world implementations. Thus we ask:

\begin{quote}
    \emph{\Gls{rq}~2: Can we map  stateful encryption schemes to \gls{aead} cipher modes to include these modes in the analysis of secure channels?}
\end{quote}

We answer \gls{rq}~2 in the affirmative: in \cref{tab:supportedaead}, we classify the input data for \modesModeled different \gls{ssh} stateful encryption schemes. We provide pseudocode for the encryption in \cref{lst:sfSshEaM,lst:sfSshEtM,lst:sfAnyXcm,lst:sfAnyCcp}. There, we show which values from \cref{tab:supportedaead} are provided via the stateful encryption interface and how the \gls{aead} input data (especially \aeadAD and~\aeadNonce) are constructed. Likewise, we show how the \gls{aead} function is called and how the state is updated. Thereby, we introduce a pseudocode translation between the stateful encryption and \gls{aead} cipher modes.

\begin{table}[t]
    \centering
    \caption{Syntax for \glsfmtshort{aead}~\cite{CCS:Rogaway02,rfc5116} and Stateful Encryption~\cite{C:JKSS12,C:KraPatWee13}.}
    \begin{tabular}{cc}
    \toprule
        Paradigm & Encryption Syntax\\
     \midrule
        \Gls{aead} & $(\cmsg, \aeadTag) \leftarrow \aeadenc(\key, \aeadNonce, \aeadAD, \msg)$ \\
    \midrule
        Stateful Enc. & $(\sfHeader, \cmsg, \st') \leftarrow \sfenc(\key, \st, \msg)$ \\
    \bottomrule
    \end{tabular}
    \label{tab:syntax}
\end{table}

\subsubsection{Formal Models and the Terrapin Attack.}

When we first analyzed the differences in resilience against the Terrapin attack among the \gls{ssh} stateful encryption schemes, we found that ChaCha20-Poly1305 is vulnerable because it has no state except for the adversary-chosen sequence numbers. But we also found that the presence of additional state does not imply security against the attack: \gls{etm} schemes are vulnerable, while \gls{eam} schemes are not. Thus, we ask:

\begin{quote}
    \emph{\Gls{rq}~3: Can we prove that SSH stateful encryption modes \emph{``unaffected''} by Terrapin are secure in our model? Can we further generalize our security model to characterize the minimum attacker capabilities required for a successful attack?}
\end{quote}

We answer it affirmatively and go beyond modeling the Terrapin attack by introducing three different \SND~oracles:

\begin{enumerate}
    \item The Terrapin attack is the strongest attack because it only assumes a weak \gls{co} adversary, which we model through a \SNDCO oracle. In the corresponding experiment, we can formally reproduce the results from the Terrapin paper with a notable difference: In our model, \gls{etm} schemes are unconditionally vulnerable, while in~\cite{USENIX:BauBriSch24}, these schemes are classified as \emph{``affected with limited exploitability''} because corrupted plaintext leads to failures at the application layer with high probability. Another difference is that we also apply the analysis of this model to \gls{etm}-Stream and \gls{eam}-Stream schemes.

    \item Standard models for secure channels~\cite{CCS:BelKohNam02,C:JKSS12,C:KraPatWee13,RSA:BHMS16,C:FGMP15,AC:PatRisShr11,C:RogZha18,EPRINT:KohPalBla03} require security under \gls{cpa} from the channel, which we model through a \SNDCPA oracle, where the adversary has full control over the plaintext. We show that in this scenario, all \gls{eam} modes become insecure. This is straightforward for \gls{eam}-\glsxtrshort{ctr} and \gls{eam}-Stream, since both are essentially stream ciphers and encryption is not protected by the \gls{mac} (\cref{subsec:cpaeam1}). For \gls{eam}-\glsxtrshort{cbc}, we present a new attack using \gls{beast}-like techniques~\cite{rizzo2011beast} to break channel integrity (\cref{subsec:cpaeam2}).

    \item In between \gls{co} and \gls{cpa}, we consider security under \gls{kpa}, which we model with a \SNDKPA oracle, where the adversary learns the plaintext but cannot choose it. We leverage this model to differentiate between the \gls{eam} modes: in this model, our new attack does not work against \gls{eam}-\glsxtrshort{cbc}, and we can prove the security of this mode, but \gls{eam}-\glsxtrshort{ctr} and \gls{eam}-Stream remain insecure.
\end{enumerate}

\subsection{Contributions}

\begin{itemize}
    \item We present the first formal security model for attacks under partially chosen channel state, which is syntax-compatible with previous work on secure \gls{tcp} channels (\cref{sec:model}). We model state manipulation by introducing \SET oracles. In the \gls{co} variant, which captures the Terrapin attack, we prove that all EtM modes and ChaCha20-Poly1305 are insecure. Going further, we prove that the \emph{``unaffected''} modes \gls{gcm}, \gls{eam}-\gls{cbc} and \gls{eam}-\gls{ctr} are in fact \gls{co}-secure. We also cover two Stream-based schemes not analyzed in~\cite{USENIX:BauBriSch24}.

    \item We extend our model to \gls{kpa} and \gls{cpa} adversaries to find the minimum attacker capabilities required to break each scheme. We prove that \gls{gcm} remains secure even against \gls{cpa}, and \gls{eam}-\gls{cbc} against \gls{kpa}, while \gls{eam}-\gls{ctr} and \gls{eam}-Stream are \gls{kpa}-insecure. For~\gls{eam}-\glsxtrshort{cbc} we give a novel \gls{cpa} attack on channel integrity, based on techniques from the \gls{beast} attack~\cite{rizzo2011beast}. This yields a strict separation between the \gls{eam} modes that the informal analysis in~\cite{USENIX:BauBriSch24} could not make (\cref{sec:analysis}).

    \item We document \modesModeled stateful encryption schemes used in \gls{ssh} in cryptographic pseudocode, making them accessible to formal analysis. We show that they are vastly different and even deviate from the \gls{aead} syntax (\cref{sec:modes}).

    \item We harmonize \gls{aead} and stateful encryption interfaces by distinguishing derived input values, randomly chosen values, channel state, associated data, and plaintext (\cref{tab:supportedaead}).
\end{itemize}

\section{Background}
\label{sec:background}
\subsection{Notation}

In this paper, we use mostly standard notation. $a \concat b$ denotes concatenation of $a$ and $b$, $\varepsilon$ is the empty word; $a, b \gets c$ is a shorthand for $a \gets c$ and $b \gets c$. $a \sample \mathcal{S}$ denotes randomly sampling an element from set~$\mathcal{S}$. $\padssh{\cdot}$ and $\padssha{\cdot}$ are the SSH padding functions described in \cref{subsec:ssh}. $\len{a}$ denotes the byte length of $a$. \lastBlock{\cmsg} returns the last cipher block of \cmsg. $\byte{ff}^i$ is the octet string that consists of $i$ repetitions of the byte \byte{ff}. $s_{i...j}$ are the $j-i+1$ bytes from index $i$ to $j$ (inclusive). $\blklen$ denotes the block length of a cipher in bits.

\subsection{\glsfmtshort{ssh}}
\label{subsec:ssh}

The core features of \gls{ssh}~2.0 are specified in \glspl{rfc}~4250--4254~\cite{rfc4250,rfc4251,rfc4252,rfc4253,rfc4254}. Several documents describe \gls{eam} modes. \Gls{rfc}~4253~\cite{rfc4253} specifies \gls{rc4}- and \gls{cbc}-based modes, \gls{rfc}~4344~\cite{rfc4344} specifies \gls{ctr}-based modes, and \gls{rfc}~4345~\cite{rfc4345} specifies additional \gls{rc4}-based modes. While \gls{rc4} has been deprecated in \gls{ssh} in \gls{rfc}~8758~\cite{rfc8758}, it may still be in use. The \gls{etm} modes are described---very briefly---in~\cite{deviationsossh}. ChaCha20-Poly1305 was added by an OpenSSH specification and was recently adopted by the \gls{ietf} sshm working group~\cite{ietf-sshm-chacha20-poly1305-04}, while \gls{gcm} is specified in \gls{rfc}~5647~\cite{rfc5647} and~\cite{miller-sshm-aes-gcm-01} (\gls{aes} only).

\subsubsection{\glsfmtshort{ssh} Key Exchange.}

\gls{rfc}~4253 describes the \gls{ssh} key exchange up to the establishment of the secure channel. This key exchange is initiated with an \gls{ascii}-over-\gls{tcp} banner exchange with no prescribed order. This is followed by two \textsc{KexInit} messages contained in unprotected \emph{\glsxtrlong{bpp}}~(\glsxtrshort{bpp}) packets; here, \gls{bpp} only adds some headers. These messages contain lists of cryptographic algorithms. Both client and server use a deterministic algorithm to derive the negotiated algorithms from these lists. Again, there is no predefined order; the client and server may proceed as they wish. The next two messages implement a \gls{dhke}, and the server adds a signature over a partial transcript of the handshake to authenticate itself. By sending \textsc{Newkeys} messages, both parties activate the secure \gls{bpp} channels using the negotiated algorithms and keys. Client authentication (\gls{rfc}~4252,~\cite{rfc4252}) is done later and is irrelevant for the Terrapin attack.

\subsubsection{\glsfmtshort{ssh} Secure Channel.}

\gls{rfc}~4253 specifies \gls{bpp} as a secure \gls{eam} channel (cf.~\cite{JC:BelNam08}): each plaintext packet consists of a length header, a padding length header, the plaintext, and random padding. The length field counts the padding length field, the payload, and the padding, but neither itself nor the \gls{mac}. The ciphertext is computed over both length fields, the payload and the padding. The \gls{mac} is computed over the sequence number followed by the same fields, and is appended in the clear. Please note that to mitigate the attack described in~\cite{SP:AlbPatWat09}, this structure has been significantly altered in different ways; for details see \autoref{sec:modes}.

\subsubsection{\glsfmtshort{ssh} Padding.}

\gls{ssh} padding plays an important role in determining the exact probabilities in some attacks and proofs. In \gls{rfc}~4253~\cite{rfc4253}, the padding of plaintext messages \msg is described for \gls{eam} modes: \msg is padded by appending $n \geq 4$ random bytes and prepending $n$ as a 1-byte padding length field, where $n$ is chosen such that $\len{\msg}+n+5$ is a multiple of $\blkAlign = \max(\sfrac{\blklen}{8}, 8)$. For stream ciphers, this yields an 8-byte alignment. We denote this process as \padssh{\msg}, whose output has length $\len{\msg}+n+1$. The remaining $4$~bytes account for the packet length field, which \padssh{\msg} does not emit but which is prepended separately and, under \gls{eam}, counted toward the alignment because it is encrypted alongside \msg. For non-\gls{eam} cipher modes, the padding differs slightly: $\len{\msg}+n+1$ is aligned instead of $\len{\msg}+n+5$. The 4-byte difference is the packet length field, which here is sent unencrypted or encrypted under a separate cipher instance and thus excluded from the alignment. We denote this variant as \padssha{\msg}.

The padding length of \gls{ssh} is inherently limited by the 1-byte padding length field and $B$. This clashes with asymptotic proofs because, while tiny, the probability of an adversary~$\adv$ guessing the padding is non-negligible in \secpar. We therefore make the reasonable assumption that $\blkAlign = \Omega(\secpar)$, corresponding to a growing $\blklen$ and minimal alignment. Consequently, $\len{\padssh{\msg}} = \Omega(\secpar)$.

\subsection{\glsfmtshort{aead} Ciphers}

Bellare and Namprempre~\cite{AC:BelNam00,JC:BelNam08} list three generic \gls{aead} construction modes. For block ciphers, there is an additional padding step, which is omitted below.

\begin{itemize}
    \item \textbf{\glsxtrfull{mte}.} In this mode, a \gls{mac}~\aeadTag is computed over the plaintext, concatenated with it, and the combination is encrypted. This is the default \gls{aead} mode in \gls{tls}~1.2.

    \item \textbf{\glsxtrfull{eam}.} The \gls{mac}~\aeadTag is computed over the plaintext, but only the plaintext is encrypted. \aeadTag is then attached to the ciphertext. This is the default mode for \gls{ssh} in~\cite{rfc4253}.

    \item \textbf{\glsxtrfull{etm}.} First, the plaintext is encrypted, then the \gls{mac}~\aeadTag is computed over the ciphertext. For \gls{ssh}, this mode is available, and it is the default mode in \gls{tls}~1.3.
\end{itemize}

Motivated by existing constructions in secure network channels, Rogaway~\cite{CCS:Rogaway02} introduced the notion of \emph{\gls{aead}}. He summarized the known generic constructions and added two novel ones: \emph{nonce stealing} and \emph{ciphertext translation}. He proposed an interface in which the input for \gls{aead} encryption is specified as a 4-tuple of a single encryption key~\key, a nonce~\aeadNonce, associated data~\aeadAD, and plaintext~\msg. The \gls{ietf} adopts this interface in~\cite{rfc5116}, and nonce stealing is widely used in specifications.

\begin{definition}[\glsfmtshort{aead}]
    An \glsxtrfull{aead} scheme is a triple $\aeadAlg = (\kgen, \enc, \dec)$ where:
    \begin{itemize}
      \item $\key \sample \kgen(\secparam)$ generates a key~$\key$ given a security parameter~$\secpar$;

      \item $\cmsg \concat \tau \sample \enc(\key, \aeadNonce, \aeadAD, \msg)$ takes a secret key~$\key$, a nonce~\aeadNonce, associated data~\aeadAD, and a plaintext~\msg, and outputs a ciphertext~$\cmsg$ and authentication tag~\aeadTag;

      \item $\msg \leftarrow \dec(\key, \aeadNonce, \aeadAD, \cmsg \concat \aeadTag)$ takes a secret key~\key, a nonce~\aeadNonce, associated data~\aeadAD, and a ciphertext~$\cmsg \concat \aeadTag$, and returns either the plaintext~\msg or a failure symbol~$\bot$.
    \end{itemize}
\end{definition}

\subsection{Stateful Encryption Schemes}

While Rogaway~\cite{CCS:Rogaway02} analyzed the construction of single network packets to derive the abstraction of \gls{aead}, Bellare, Kohno, and Namprempre~\cite{CCS:BelKohNam02} analyzed---using \gls{ssh} as an example---how the \emph{sequence} of network packets can be protected from replay, reordering, and packet deletion attacks. In practice, this is achieved by using sequence numbers in \gls{aead}. The authors formalized such constructions as \emph{stateful encryption}. In these constructions, encryptions of subsequent packets depend not only on the key, randomization, and plaintext, but also on some \emph{state} inherited from the previous packet encryption.

\begin{definition}[Stateful Encryption Scheme]
    A \emph{stateful encryption scheme} is a triple $\sfAlg = (\init, \enc, \dec)$ where:
    \begin{itemize}
      \item $(\key, \st[0]) \sample \init(\secparam)$ generates a key~$\key$ and an initial state~$\st[0]$ given a security parameter~$\secpar$;

      \item $(\sfHeader, \cmsg, \st[i+1]) \sample \enc(\key, \st[i], \msg)$ takes a secret key~$\key$, the state~\st[i], and a message~\msg, and outputs a header~\sfHeader, a ciphertext~\cmsg, and an updated state~\st[i+1];

      \item $(\msg, \st[j+1]) \leftarrow \dec(\key, \st[j], \sfHeader, \cmsg)$ takes a secret key~$\key$, the state~\st[j], a header~\sfHeader, and a ciphertext~\cmsg, and outputs a message~\msg (or~$\bot$ in case of failure), and an updated state~\st[j+1].
    \end{itemize}
    W.l.o.g., each state is partitioned as~$\st = (\advst, \intst)$ where $\advst \in \mathcal{S}_\mathsf{adv}$ is the adversary-chosen state and $\intst \in \mathcal{S}_\mathsf{sec}$ is the secure state.
\end{definition}

The abstraction of stateful encryption has been used to analyze the security of protocols~\cite{C:JKSS12,C:KraPatWee13,C:FGMP15} and to develop theoretical constructions~\cite{RSA:BHMS16}.

\subsection{Terrapin Attack}
\label{sec:bg:terrapin}

The Terrapin attack~\cite{USENIX:BauBriSch24} is illustrated in \cref{fig:terrapin}. It exploits two design flaws of the \gls{ssh} Transport Layer Protocol: first, the server's signature in the second of the two key exchange messages does not cover the whole handshake transcript, but only a subset thereof, allowing injection of optional messages (e.g., \textsc{Ignore}) that may be sent at any time during the key exchange; and third, in contrast to \gls{tls}, \gls{ssh} sequence numbers are not reset when the key exchange concludes.

\begin{figure}[t]
    \centering
    \includegraphics[width=0.65\columnwidth]{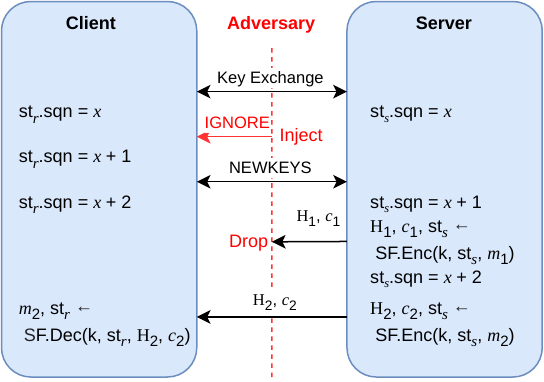}
    \caption{Illustration of the Terrapin attack on the \glsfmtshort{bpp} from~\cite{USENIX:BauBriSch24}. The server sends $m_1$ and $m_2$, but the client receives only $m_2$.}
    \label{fig:terrapin}
\end{figure}

This enables the \emph{prefix truncation attack} on the secure channel of \gls{ssh} depicted in \cref{fig:terrapin} (with $x = 2$): The \gls{mitm} adversary injects an optional and ignored message, such as \textsc{Ignore}, into the unprotected \gls{bpp}---since no encryption or message integrity is yet active, the resulting change in the client's receiving sequence number from 2 to 3 remains undetected. After receiving the legitimate \textsc{Newkeys} message from the server, the value is incremented to 4, which is then used to validate the integrity of the first encrypted message received by the client. When the server sends the \textsc{Newkeys} message, it increments the corresponding sending sequence number to 3 and, after sending $(\sfHeader_1, \cmsg_1)$, to 4. The \gls{mitm} adversary then drops $(\sfHeader_1, \cmsg_1)$ during transmission; this is the aforementioned \emph{prefix truncation}. During the stateful encryption of $(\sfHeader_2, \cmsg_2)$, the server uses sequence number 4 as part of the channel state for the server-to-client unidirectional channel. When receiving $(\sfHeader_2, \cmsg_2)$, the client uses the same value to perform the stateful decryption, which consequently succeeds, leaving the adversary's prefix truncation undetected by the recipient.

\section{Model}
\label{sec:model}
We define a model for \emph{channel integrity under partially chosen state attacks} (\intpcsa). Channel integrity addresses both forged \gls{aead} ciphertexts and reordering of messages within the channel. We formally define the security of a stateful encryption scheme as a security game played between a \gls{ppt} adversary $\adv$ and a challenger $\cdv$ that models a single unidirectional \gls{bpp} channel described as a stateful encryption scheme. The capabilities of the adversary are specified in \cref{lst:model} as pseudocode oracles. The \SND~oracle models the sender side of the channel and the \RCV~oracle models the receiver; the sender and receiver have already agreed on a common key~$\key$ and initial channel state $\st[s] = \st[r]$. The adversary may query these oracles with any inputs and receive the output values; this models that the adversary has a \gls{mitm} position. In addition to these \emph{network privileges}, $\adv$ gets one \emph{plaintext privilege} (ciphertext-only/\gls{co}, known plaintext/\gls{kpa}, chosen plaintext/\gls{cpa}) through the different \SNDX~instantiations, where $\X \in \{\CO, \KPA, \CPA\}$. Finally, $\adv$ gets \emph{partially chosen state privilege} through the \SET~oracles.

The goal of $\adv$ is to compute the challenge bit~\challbit; if its result \challbitp equals \challbit, $\adv$ \emph{wins} the game. Since a randomly chosen bit \challbit can be guessed with probability $\sfrac{1}{2}$, winning the security game does not immediately render a stateful encryption scheme insecure; therefore, we eliminate the effect of guessing by measuring the absolute difference between the winning probability of $\adv$ and the guessing probability $\sfrac{1}{2}$ and call this the \emph{advantage} of $\adv$. Since \challbit is only used in \RCV, the adversary can only increase its advantage by calling \RCV.

\begin{figure*}
    \input{img/cryptocode/int-pcsa}
    \caption{Our model for channel integrity under partially chosen state attacks. The pseudocode defines a single-instance security game played between an adversary $\adv$ and a \intpcsa challenger, where the adversary has access to both \SET~oracles, the \RCV~oracle, and one of the \SND~oracles. The different \SND~oracles model the adversary's knowledge of and control over the plaintext: \glsfmtfull{co}, known plaintext (\glsfmtshort{kpa}), or chosen plaintext (\glsfmtshort{cpa}).
    }
    \label{lst:model}
\end{figure*}

\subsection{Experiment}

\subsubsection{Exp.}

During initialization, $\cdv$ calls \sfsetup to generate a symmetric key~$\key$ uniformly at random alongside an initial state~$\st[s]$.\footnote{In real-world protocols, an \glsxtrfull{ake} protocol is used for agreement on a key~$\key$ and possibly an initial state~$\st[s]$, for example, the key exchange in the \gls{ssh} Transport Layer Protocol.} By definition, $\st[s]$ consists of an adversary-chosen state~$\advst$ and secure state~$\intst$. $\cdv$ then duplicates the initial state for the \RCV oracle as~$\st[r]$, thereby modeling synchronized initial encryption and decryption states. The structure of the state depends on the stateful encryption scheme used; in all instantiations considered in this paper, the sequence number \sqn forms the adversary-chosen state and initialized to zero (cf. \cref{tab:supportedaead}).

Besides that, we need bookkeeping variables. First, there is the challenge bit~\challbit. This bit is an artifact used to compute the advantage of the adversary in the security game and does not occur in real-world instantiations. $\cdv$ also initializes the sent and received counters, \idx[s] and \idx[r] (not to be confused with sequence numbers), which point to entries in an array \msgStore, to zero. The array \msgStore is used by $\cdv$ to track the adversary's queries to \SND and to check the correct order of decryption in \RCV. Every entry in \msgStore is initialized just-in-time; accessing an uninitialized array index returns $\bot$. We do not allow adaptive calls to \SET in our model, which is implemented by the flags \setflag[s] and \setflag[r]. The adversary can set a corresponding adversary-chosen state only before the first call to \SND or \RCV has occurred; when \SND or \RCV is called, we block the adversary from using the corresponding \SET oracle by setting \setflag[s] or \setflag[r]. After the initialization, the adversary $\adv$ is provided access to one of the three \SND oracles, as well as \RCV and both \SET oracles. After polynomially many steps, the adversary returns a bit~\challbitp.  The adversary wins the game if $\challbitp=\challbit$.

\begin{definition}
    For $\X \in \{\CO, \KPA, \CPA\}$, a stateful encryption scheme \sfAlg is \intpcsax-secure if, for every \gls{ppt} adversary~$\adv$, the advantage \[\advantage{\intpcsax}{\adv} =\\ \left|\prob{\challbitp=\challbit} - \frac{1}{2}\right|\] is negligible in~$\lambda$.
\end{definition}

\subsubsection{SND.}

\SND comes in three different flavors. \SNDCO allows the adversary to choose the byte length $\ell$ of the message and outputs the authenticated header~\sfHeader and ciphertext~\cmsg of a randomly chosen message~\msg. \SNDKPA returns the randomly chosen message~\msg in addition to $(\sfHeader, \cmsg)$. Finally, \SNDCPA allows $\adv$ to choose \msg as in a \gls{cpa} oracle. In each case, the message \msg is encrypted by calling the stateful encryption scheme; afterward, the state variable \st[s] is updated accordingly. The stateful encryption scheme can be instantiated with each of the \modesModeled schemes from \cref{tab:supportedaead} by using the appropriate pseudocode interface from one of the \cref{lst:sfSshEaM,lst:sfSshEtM,lst:sfAnyXcm,lst:sfAnyCcp}. If the stateful encryption succeeds, we store the plaintext \msg in the next free slot of the \emph{sent array} \msgStore, indexed by \idx[s]. Afterward, we increase the sent counter \idx[s], and finally return (\sfHeader, \cmsg) to the adversary $\adv$.

\subsubsection{RCV.}

\RCV is a decryption oracle that accepts chosen header data~\sfHeader and any ciphertext~\cmsg as input. \cmsg can be a ciphertext previously issued by \SND---in which case \gls{aead} decryption will most likely succeed---or a ciphertext generated or modified by $\adv$---in which case stateful decryption will most likely return an error symbol $\bot$ caused by a \gls{mac}/tag verification failure. Because of the \textbf{else} condition on line 7, in most cases, the value $\bot$ will be returned. The only exception is when \cmsg decrypts to a novel plaintext or when decryption happens out of order. This is checked by the \textbf{if}-clause in line 5. The \RCV oracle checks whether the plaintext returned on the ($\idx[r]+1$)-th successful decryption call is \emph{not} equal to the plaintext entry in \msgStore at index \idx[r].

By definition of the security game, \gls{mac} forgeries or out-of-order deliveries are not \emph{directly} considered winning events. However, when such an event happens and $\challbit = 1$, $\adv$ gets precise information about \challbit and can thus maximize its advantage. Whenever a plaintext $\msg \neq \bot$ is returned, $\adv$ \emph{knows} that $\challbit=1$ and can set $\challbitp=1$. This is why the last, seemingly strange, check is implemented in line 10 of \RCV. The benefit of this advantage-based security notion is that our game can easily be extended to a full game for channel security, also covering $\indccaii$ security (cf.\ \cref{lst:model2} in the Appendix), without changing the syntax of the model. This follows a line of research in channel security initiated by~\cite{AC:PatRisShr11,C:JKSS12,C:KraPatWee13}.

\subsection{Restrictions}
\label{subsec:restrictions}

\subsubsection{Application Layer.}

We do not model the application-layer protocols to which the plaintexts of the secure channel are delivered. Any novel or out-of-order plaintext $\msg'$ accepted by \RCV constitutes a successful attack, regardless of whether it is a valid application-layer message. This is why some ``unsuccessful'' attacks on \gls{ssh} in the Terrapin paper are successful in our model---the Terrapin authors \emph{did} take the application layer of \gls{ssh} into account.

\subsubsection{IND-CCA2.}

Comparing \cref{lst:model} and \cref{lst:model2}, it is clear that we are \emph{not} considering $\indccaii$ security in this paper. The reasons are twofold:
\begin{enumerate}
    \item Our goal is to determine unambiguously whether an \gls{ssh} cipher mode is \emph{secure} against the Terrapin attack. Including $\indccaii$ security in our security model would make this determination no longer possible. For example, \gls{eam}-CBC is secure against the Terrapin attack but not $\indccaii$-secure, as shown in~\cite{SP:AlbPatWat09}. Consequently, we would have to label it as ``insecure''.

    \item The $\indccaii$ status of some cipher modes is not yet known. While some results have been presented in~\cite{EC:PatWat10} and~\cite{CCS:ADHP16}, not all modes are covered yet. This would have required the category ``unknown'' in the classification.
\end{enumerate}

\subsection{Differences to Formal Models in Related Work}
\label{subsec:differences}

Bellare, Kohno, and Namprempre~\cite{CCS:BelKohNam02} model the indistinguishability of ciphertexts and the integrity of plaintexts/ciphertexts in two separate security games. We use an idea from Paterson, Ristenpart, and Shrimpton~\cite{AC:PatRisShr11}, who combine the two parts into one advantage-based experiment (\cref{lst:model2}), except that we omit the \indccaii\xspace component in \cref{lst:model} to focus on the Terrapin attack.

We do not follow~\cite{AC:PatRisShr11} in allowing the adversary to control the amount of padding, since many libraries do not expose this functionality to the application layer. Thus, the $\ell$ parameter has a different meaning in our model---it specifies the length of a plaintext, not of the ciphertext.

Several models use different levels of checks on the ordering of messages. Rogaway and Zhang~\cite{C:RogZha18}, Boyd et al.~\cite{RSA:BHMS16}, and Kohno, Palacio, and Black~\cite{EPRINT:KohPalBla03} define four levels: $L_1$ (no message injections), $L_2$ (no replays), $L_3$ (no reordering), and $L_4$ (receiving order must match sending order). A fifth level, $L_5$, defined in the appendix of~\cite{EPRINT:KohPalBla03}, differs from $L_4$ only in the ability of the adversary to send decrypt queries after a failed attempt, so $L_5$ matches theoretical models~\cite{AC:PatRisShr11}, while $L_4$ better describes \gls{ssh} in practice. As our goal is to determine whether a given stateful encryption scheme is vulnerable to Terrapin, we implement $L_4$ in our model. By removing the failure flag $\rcvfailflag$, the model can be converted to an $L_5$ model. In~\cite{C:RogZha18}, a new approach is proposed to unify security models for stateful encryption. This is motivated by the complex and often counterintuitive security definitions for the different channel security levels presented in the papers cited above. Rogaway and Zhang use correctness definitions of the different channels to define the notion of \emph{indistinguishability up to correctness}. We hope that by focusing exclusively on \gls{tcp}-based secure channels, the complexity of our model remains manageable.

Jager~et~al.~\cite{C:JKSS12} use a multi-instance stateful variant of the model from~\cite{AC:PatRisShr11} to model \gls{acce}. This is justified because in \gls{tls}, many stateful encryption channels are instantiated in parallel and an active adversary may exploit this.

\subsection{Additional Related Work}
\label{sec:rw}
\textbf{\glsfmtshort{aead}.} Katz and Yung~\cite{FSE:KatYun00} proposed one of the first \gls{aead} modes. Generic constructions for \gls{aead} were investigated by Krawczyk~\cite{C:Krawczyk01}. The IND-CCA3 definition from Shrimpton~\cite{EPRINT:Shrimpton04} is a predecessor of the unified security definition for \gls{aead} given by Paterson, Ristenpart, and Shrimpton~\cite{AC:PatRisShr11}. \textbf{\glsfmtshort{acce}.} In the analysis of all three cipher suite families from \gls{tls}~1.2, Krawczyk, Paterson, and Wee~\cite{C:KraPatWee13} and Jager~et~al.~\cite{C:JKSS12,JC:JKSS17} use the \gls{acce} model, which combines authenticated key agreement with stateful \gls{aead}. \textbf{Data Streams.} Fischlin~et~al.~\cite{C:FGMP15} refined the stateful channel security model further to the data \emph{streams} sent and received by the application layer, which would be partitioned into packets solely by the sender's buffer. They gave the adversary access to a \emph{flush} interface, allowing the adversary to artificially partition the application data stream into packets. \textbf{Constructions.} The above works provide abstract formal security games to analyze the security of existing real-world implementations of protocols like \gls{ssh} and \gls{tls}. Attacks on these protocols revealed that the abstract security model did not cover all subtleties of the protocol implementations. Subsequently, models were adapted to the protocol implementations (\gls{tls}: Paterson, Ristenpart, and Shrimpton~\cite{AC:PatRisShr11}; \gls{ssh}: Albrecht, Paterson, and Watson~\cite{SP:AlbPatWat09}, as well as Paterson and Watson~\cite{EC:PatWat10}), or implementations were adapted to models after attacks (\gls{ssh}: Bäumer, Brinkmann, and Schwenk~\cite{USENIX:BauBriSch24}, Miller~\cite{ietf-sshm-strict-kex-01}). Delignat-Lavaud et. al \cite{SP:DFKPRS17} modeled different TLS 1.3 AEAD modes; their model does not include channel state.

\section{Channel State and \glsfmtshort{aead} Cipher Modes}
\label{sec:modes}
\begin{table*}[t]
    \renewcommand{\crefpairconjunction}{, }
    \centering
    \begin{threeparttable}
    \caption{Supported stateful encryption schemes and their parameters in \glsfmtshort{ssh}. The columns \glsfmtshort{co}, \glsfmtshort{kpa}, and \glsfmtshort{cpa} indicate if a given scheme is secure in the corresponding security model. In \glsfmtshort{ssh}, each connection consists of two unidirectional channels, each with a separate state and keying material. The adversary-chosen state component is~$\advst$, and the secure state component is~$\intst$. For \glsfmtshort{gcm}, \sqn is maintained but not used.}
    \label{tab:supportedaead}
    \begin{tabular}{l>{\centering\columncolor{gray!20}}p{\widthof{unaffected}}ccccc>{\columncolor{blue!15}}c>{\columncolor{blue!15}}c>{\columncolor{blue!15}}cc}
    \toprule
    \thead[l]{SSH Cipher} & \thead[b]{\cite{USENIX:BauBriSch24}\\(informal)} & \thead{$\sfHeader$} & \thead{\st.\advst} & \thead{\st.\intst} & \thead[l]{Key Derivation} & \thead{Fig.} & \thead{\glsfmtshort{co}} & \thead{\glsfmtshort{kpa}} & \thead{\glsfmtshort{cpa}} & \thead{Proofs}\\[4pt]
     \hline\rule{0pt}{3ex}
     \gls{eam}-\gls{cbc} &unaffected&  $-$ & \sqn & \IV  & \kenc, \kauth, \IV[\mathsf{init}] & \labelcref{lst:sfSshEaM,fig:sfSshEaM} & \yes & \yes & \no & \labelcref{subsec:cpaeam2,subsec:cpaeam4} \\
     \gls{eam}-\gls{ctr} &unaffected& $-$ & \sqn & \ctr  & \kenc, \kauth, \ctr[\mathsf{init}] & \labelcref{lst:sfSshEaM,fig:sfSshEaM} & \yes & \no & \no & \labelcref{subsec:cpaeam1,subsec:cpaeam3} \\
     \gls{eam}-Stream & $-$ & $-$ & \sqn & \pos & \kenc, \kauth & \labelcref{lst:sfSshEaM,fig:sfSshEaM} & \yes\tnote{1} & \no & \no & \labelcref{subsec:cpaeam1,subsec:cpaeam3} \\[4pt]
     \gls{etm}-\gls{cbc} &affected& \len{m'} & \sqn & \IV & \kenc, \kauth, \IV[\mathsf{init}] & \labelcref{lst:sfSshEtM,fig:sfSshEtM} & \no & \no & \no & \labelcref{subsec:etmattack} \\
     \gls{etm}-\gls{ctr} &affected& \len{m'} & \sqn & \ctr & \kenc, \kauth, \ctr[\mathsf{init}] & \labelcref{lst:sfSshEtM,fig:sfSshEtM} & \no & \no & \no & \labelcref{subsec:etmattack}  \\
     \gls{etm}-Stream & $-$ & \len{m'} & \sqn & \pos  & \kenc, \kauth & \labelcref{lst:sfSshEtM,fig:sfSshEtM} & \no  & \no  & \no & \labelcref{subsec:etmattack}  \\[4pt]
     \gls{gcm} & unaffected& \len{m'} & \color{gray}\sqn & \ictr & \key, \salt, \ictr[\mathsf{init}] & \labelcref{lst:sfAnyXcm,fig:sfSshGcm} & \yes & \yes & \yes & \labelcref{subsec:gcmproof} \\
     CCP\tnote{2} &affected& $-$ & \sqn & $-$ & $\key = (\key_2, \key_1)$ & \labelcref{lst:sfAnyCcp,fig:ccp} & \no & \no & \no & \labelcref{subsec:genattack} \\
     \bottomrule
    \end{tabular}
    \begin{tablenotes}
        \item [1] Unknown for RC4
        \item [2] ChaCha20-Poly1305 as stateful encryption, SSH variant
    \end{tablenotes}
    \end{threeparttable}
\end{table*}

In this section, we analyze the stateful encryption schemes supported by \gls{ssh} and how they construct the \gls{aead} input \aeadNonce and \aeadAD from the message~\msg and the channel state~\st. Furthermore, we detail the structure of the authenticated header~\sfHeader. An overview is given in \cref{tab:supportedaead}.

\subsection{Preliminaries}

\subsubsection{Associated Data and Channel State.}

As \cref{tab:supportedaead} shows, for all described stateful encryption schemes in SSH, \emph{channel state} always consists of the sequence numbers as an adversary-chosen state component and may have a non-empty secure state component depending on the underlying \gls{aead} mode. Thus, channel state may influence \gls{mac} validation directly (when included in \aeadAD) or indirectly (for \intptxt security, by changing the decryption). The inclusion method may differ; e.g., a sequence number may either be concatenated with the packet header (\gls{eam}, \gls{etm}), or used with nonce stealing (ChaCha20-Poly1305).

Historically, after the dissolution of the original secsh working group and before the formation of the sshm working group by the \gls{ietf}, \gls{ssh} development was primarily driven by OpenSSH. Of the eight commonly supported stateful encryption schemes in \gls{ssh}, only half are described in \glspl{rfc}; these include the three generic modes using \gls{eam} constructions described in~\cite{rfc4253,rfc4344}. The other half, including the three generic modes using \gls{etm} constructions, are described in separate documents~\cite{deviationsossh,ietf-sshm-chacha20-poly1305-04}.

\subsection{\glsfmtlong{eam}}

The structure of this encryption mode is illustrated in \cref{fig:sfSshEaM}, and the pseudocode translation from stateful encryption to \gls{aead} in \cref{lst:sfSshEaM}.

\begin{figure}
    \centering
    \includegraphics[width=0.7\columnwidth]{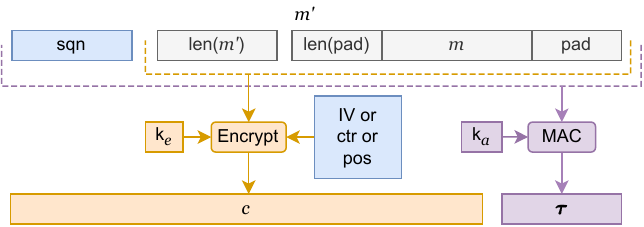}
    \caption{Encryption of a binary packet in \glsfmtshort{ssh} using a generic \glsfmtlong{eam} construction~\cite{rfc4253}.}
    \label{fig:sfSshEaM}
\end{figure}

\begin{figure*}
    \include{img/cryptocode/sf-ssh-eam}
    \caption{Stateful encryption schemes using a generic \glsfmtshort{eam} construction in \glsfmtshort{ssh}~\cite{rfc4253}. For block ciphers in \glsfmtshort{cbc} mode, the \glsfmtshort{iv} becomes the secure state because \glsfmtshort{ssh} uses \glsfmtshort{iv} chaining. For block ciphers in \glsfmtshort{ctr} mode, the counter counts the input blocks of the underlying block cipher and therefore becomes secure state. For stream ciphers, the channel state includes the internal cipher state of the stream cipher, that is, the position in the keystream. The only associated data is the sequence number~\sqn as the adversary-chosen state component, so the header~$\sfHeader$ is empty.}
    \label{lst:sfSshEaM}
\end{figure*}

\subsubsection{\glsfmtshort{cbc}.}

In \gls{ssh}, when using the \gls{eam} construction with a \gls{cbc} block cipher (\cite[Sec. 6.3]{rfc4253}), the encoded plaintext $\msg'$ consists of the packet length, the padding length, the message $\msg$, and the random padding bytes. Consequently, the header~$\sfHeader$ returned by the stateful encryption is empty.

\gls{ssh} uses \gls{iv} chaining: the initial value $\IV[\mathsf{init}]$ is generated during key derivation, and after each encryption $\IV$ is set to the last block of the previous ciphertext. The channel state thus consists of the sequence number as adversary-chosen and the \gls{iv} as secure state. The \gls{mac}~$\aeadTag$ is computed over the encoded plaintext $\msg'$, prepended with the sequence number $\sqn$.

If the wrong sequence number is used for decryption, \gls{mac} verification and thus the stateful decryption operation overall will fail. If the wrong \gls{iv} is used for stateful decryption, the plaintext produced by the block cipher will change. Since the \gls{mac} is computed over the plaintext, \gls{mac} verification, and thus stateful decryption, will also fail.

\subsubsection{\glsfmtshort{ctr}.}

In \gls{ssh}, when using the \gls{eam} construction with a \gls{ctr} block cipher (\cite[Sec. 4]{rfc4344}), the counter used for encryption has the size of the block length \blklen (128 bits for \gls{aes}). Since the counter is never reset under the same key material, it forms the secure state component. The initial value of the counter is generated through key derivation similar to the initial value of the \gls{iv} in \gls{cbc} mode, that is, $\ctr[\mathsf{init}] = \IV[\mathsf{init}]$. After each block, the corresponding counter is incremented.

\subsubsection{Stream.}

\gls{ssh} supports three algorithms based on the \gls{rc4} stream cipher: \texttt{arcfour}~\cite{rfc4253} is classical \gls{rc4} with a key size of 128 bits; \texttt{arcfour128} and \texttt{arcfour256}~\cite{rfc4345} are variants with key sizes of 128 and 256 bits; additionally, the first 1536 bytes of output from the \gls{prg} are dropped to avoid the bias in these bytes. All of these modes can be used in \gls{eam} mode~\cite[Sec. 6.3]{rfc4253}, with only minor adjustments to account for the different channel state. In our formal description, the latter two modes can be described by initializing $\pos$ to $1536$ rather than zero, effectively skipping the first 1536 bytes. Although \gls{rc4} as a stream cipher does not require padding, the $\padssh{\msg}$ function still pads $\msg'$ to a multiple of 8 bytes, as mandated by the \gls{ssh} specification~\cite[Sec. 6]{rfc4253}. The stream cipher is never re-initialized for the lifetime of the key material, meaning that the current keystream position \pos becomes part of the channel state.

\subsection{\glsfmtlong{etm}}

\Gls{ssh} implementations widely support \gls{etm} as a vendor extension by OpenSSH. Here, the packet length is no longer encrypted but is instead, together with the sequence number, authenticated as part of the \gls{mac} input. The structure of this encryption mode is illustrated in \cref{fig:sfSshEtM}, and the pseudocode translation from stateful encryption to \gls{aead} in \cref{lst:sfSshEtM}. The plaintext $\msg'$ to be encrypted consists of the padding length field, the payload \msg, and the random padding.

The \gls{mac}~$\tau$ is computed over the concatenation of the sequence number \sqn, the packet length $\len{\msg'}$, and the ciphertext \cmsg. These modifications are briefly described in~\cite[Sec. 1.5]{deviationsossh}, while other aspects remain unchanged. For example, \gls{iv} chaining is used in \gls{cbc} mode and the counter is never reset in \gls{ctr} mode. Thus, the channel state for each mode is identical to that of its \gls{eam} counterpart.

\begin{figure}
    \centering
    \includegraphics[width=0.7\columnwidth]{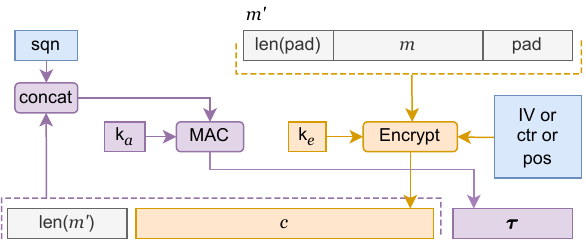}
    \caption{Encryption of a binary packet in \glsfmtshort{ssh}, using generic \glsfmtlong{etm}~\cite{deviationsossh}.}
    \label{fig:sfSshEtM}
\end{figure}

\begin{figure*}
    \include{img/cryptocode/sf-ssh-etm}
    \caption{Stateful encryption schemes using a generic \glsfmtshort{etm} construction in \glsfmtshort{ssh}~\cite{deviationsossh}. For block ciphers in \glsfmtshort{cbc} mode, the \glsfmtshort{iv} is part of the channel state because \glsfmtshort{ssh} uses \glsfmtshort{iv} chaining. For block ciphers in \glsfmtshort{ctr} mode, the counter counts the input blocks of the underlying block cipher and is therefore part of the channel state. For stream ciphers, the channel state includes the internal cipher state of the stream cipher, that is, the position in the keystream.}
    \label{lst:sfSshEtM}
\end{figure*}

\subsection{\glsfmtlong{gcm}}
\label{subsec:055_aes-gcm}

According to \gls{rfc}~5647~\cite{rfc5647}, in \gls{ssh} the \gls{aead} mode \gls{gcm} requires a secret key~$\key$, a nonce~\aeadNonce (which defines the start value of the internal counter), associated data~\aeadAD, and plaintext \msg. According to the general \gls{ssh} rules, the plaintext must be padded to a multiple of the block length of \gls{aes}. The output is a sequence of ciphertext blocks~\cmsg and an authentication tag~\aeadTag. The nonce~\aeadNonce consists of 4~fixed bytes~\salt computed during key derivation and 8~bytes for the invocation counter~\ictr; the remaining 4~bytes are reserved for the block counter, which starts at 1 (\cref{fig:sfSshGcm}, Fig.~\ref{lst:sfAnyXcm} line~2). \aeadAD includes \sfHeader, i.e., the packet length of the binary packet. In this mode, the packet length is not encrypted, see \cref{fig:sfSshGcm}. The channel state includes the implicit sequence numbers; however, \gls{ssh} with \gls{gcm} does not use the sequence numbers in any cryptographic computation (\cref{fig:sfSshGcm,lst:sfAnyXcm}); they are still incremented to allow their use with another \gls{aead} mode after a possible key re-exchange. Instead, \gls{ssh} uses the 8-byte invocation counter~\ictr that is incremented after each encryption of a packet.

\begin{figure}
    \input{img/cryptocode/sf-any-gcm}
    \caption{Stateful encryption scheme using \glsfmtshort{gcm} in \glsfmtshort{ssh}~\cite{rfc5647,miller-sshm-aes-gcm-01}. The sequence number is maintained but not used, denoted in gray color.}
    \label{lst:sfAnyXcm}
\end{figure}

\begin{figure}
    \centering
    \includegraphics[width=0.7\columnwidth]{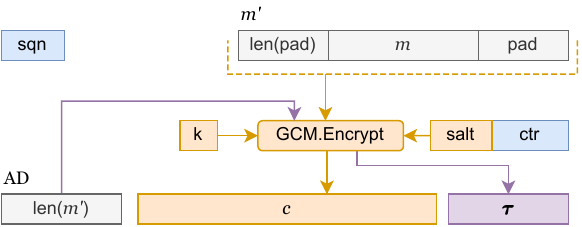}
    \caption{Encryption of a binary packet in \gls{ssh} using \glsfmtshort{gcm}~\cite{rfc5647,miller-sshm-aes-gcm-01}.}
    \label{fig:sfSshGcm}
\end{figure}

\subsection{ChaCha20-Poly1305}
\label{subsec:057_chacha20-poly1305}

ChaCha20-Poly1305~\cite{rfc8439} is a combination of the ChaCha20 stream cipher (re-initialized for each binary packet) and the \gls{mac} Poly1305. ChaCha20, like its predecessor Salsa20, generates a keystream consisting of 512-bit blocks from a $4\times 4$ matrix of 32-bit words using add-XOR-rotate operations. To produce one block of keystream, the ChaCha20 matrix is initialized with a 256-bit key~$\kenc$, a 128-bit constant, a 64-bit nonce~$\aeadNonce$, and a 64-bit block counter~\ctr initialized to 0. An \gls{ietf}~\cite{rfc8439} variant with a 96-bit nonce and a 32-bit block counter exists, but is not implemented by \gls{ssh}.

In an attempt to reconcile the \gls{ssh} \gls{bpp} specification~\cite{rfc4253} with the attack presented in~\cite{SP:AlbPatWat09}, the design of ChaCha20-Poly1305 for \gls{ssh} incorporates a separately keyed ChaCha20 instance to encrypt the packet length field. This difference is illustrated in \cref{fig:ccp}. The left-zero-padded 4-byte \gls{ssh} sequence number is used directly as a nonce.

ChaCha20 is first initialized with $\key_2$ to produce a single 512-bit block, of which 4~bytes are used to encrypt the length field via an XOR operation. Since a different keystream is used here, this mitigates the attack from~\cite{SP:AlbPatWat09}, while also allowing the packet length to be encrypted. In the second step, ChaCha20 is used in combination with Poly1305 as an \gls{aead} cipher. The previously encrypted 4-byte length field is used as associated data~\aeadAD.

\begin{figure}
    \input{img/cryptocode/sf-any-ccp}
    \caption{Stateful encryption scheme using ChaCha20-Poly1305 in \glsfmtshort{ssh}. A separately keyed ChaCha20 instance encrypts the packet length to comply with \glsfmtshort{rfc}~4253~\cite{rfc4253}, distinguishing it from the \glsfmtshort{etm} and \glsfmtshort{gcm} schemes.}
    \label{lst:sfAnyCcp}
\end{figure}

\begin{figure}
    \centering
    \includegraphics[width=0.7\columnwidth]{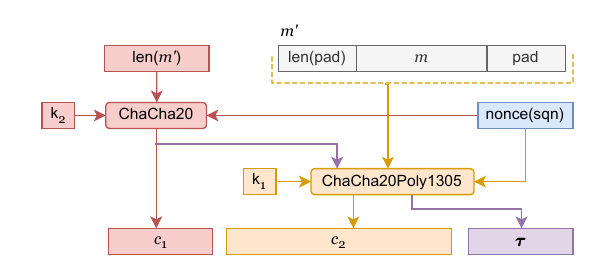}
    \caption{Encryption of one binary packet using ChaCha20-Poly1305 in \glsfmtshort{ssh}~\cite{ietf-sshm-chacha20-poly1305-04}. The construction requires a 512-bit key that is split into two 256-bit keys $\key_1$ and $\key_2$.}
    \label{fig:ccp}
\end{figure}

\section{Security Analysis}
\label{sec:analysis}
\emph{All} models for secure channels from related work (\cref{subsec:differences,sec:rw}) provide the adversary with chosen-plaintext privileges. However, Terrapin is not a \gls{cpa} attack---with chosen-plaintext privileges, schemes secure against Terrapin become vulnerable. We therefore analyze the security of the \gls{ssh} stateful encryption schemes in three slightly different security models, which are selected by choosing one of the three oracles \SNDCO, \SNDKPA, \SNDCPA.

The weakest security model, and the one modeling the Terrapin attack, is the one using \SNDCO. If an \gls{aead} scheme is insecure in this model, it is automatically insecure in the two stronger models. This is the case for the three \gls{etm} variants and for ChaCha20-Poly1305. Therefore, proving insecurity in the weakest model for these schemes as done in \cref{subsec:genattack,subsec:etmattack} is sufficient.

The strongest of our security models is the \gls{cpa} model. If we can prove that a cipher is secure in this model, then it is automatically secure in the two weaker models. This is the case for \gls{gcm}, and the proof in \cref{subsec:gcmproof} is straightforward: \gls{gcm} does not use sequence numbers at all but instead \emph{invocation counters} which are (likely) perfectly aligned with our bookkeeping indices $\idx[s], \idx[r]$.

The most interesting class of stateful encryption schemes is \gls{eam}. Here, it is relatively straightforward to show that \gls{eam}-\gls{ctr} and \gls{eam}-Stream are insecure in the \gls{kpa} setting (\cref{subsec:cpaeam1}), which also implies insecurity in the \gls{cpa} setting. A proof of their security in the \gls{co} setting, which depends on the security of the \gls{prg} used to generate the keystream, can be found in \cref{subsec:cpaeam3}. As \gls{rc4}'s \gls{prg} may be distinguishable from a random one~~\cite{FSE:ManSha01}, all concrete instantiations of \gls{eam}-Stream remain unproven. This leaves us with \gls{eam}-\gls{cbc}, the most interesting mode. In \cref{subsec:cpaeam2}, we show that this mode is \emph{not} \intpcsacpa-secure. The attack methodology is similar to the \gls{beast} attack on \gls{tls}~\cite{rizzo2011beast}, but has a different goal: instead of breaking confidentiality, we attempt to control $\st[r].\IV$ of the channel state. Finally, in \cref{subsec:cpaeam4} we show that \gls{eam}-\gls{cbc} is \intpcsakpa-secure: breaking \gls{kpa} security is equivalent to either forging a \gls{mac} or inverting the block cipher without the key.

\subsection{Insecurity of ChaCha20-Poly1305 under \glsfmtshort{co}}
\label{subsec:genattack}

In the security game described in \cref{lst:model}, stateful encryption schemes whose secure state is empty are vulnerable to a generic attack if their adversary-chosen state is predictable. In the context of \gls{ssh}, the attack described below applies to ChaCha20-Poly1305 (cf. \cref{tab:supportedaead}).

\begin{definition}[Deterministic adversary-chosen state]
    Let $\sfAlg$ be a stateful encryption scheme with state $\st_i = (\advst_i,\intst_i)$ after the $i$-th invocation of \sfenc. We say that $\sfAlg$ has a \emph{deterministic adversary-chosen state} if there exists a publicly known deterministic polynomial-time function $f_{\mathsf{adv}} : \mathcal{S}_{\mathsf{adv}} \to \mathcal{S}_{\mathsf{adv}}$ such that, for every call to \sfenc and every $i \geq 0$, $\advst_{i+1} = f_{\mathsf{adv}}(\advst_i).$ Thus, given $\advst_i$, an adversary can efficiently determine all subsequent adversary-chosen states without the need to set them first.
\end{definition}

\begin{theorem}
    \label{th:onlyext}
    Let \sfAlg be a stateful encryption scheme where the adversary-chosen state~$\advst$ is deterministic and $\intst = \varepsilon$. Then \sfAlg is \emph{not} \intpcsaco-secure.
\end{theorem}

\begin{proof}
We construct a \gls{ppt} adversary~$\adv$ against the \intpcsaco game of \cref{lst:model} that proceeds as follows:
\begin{enumerate}
    \item $\adv$ samples $\advst \sample \mathcal{S}_\mathsf{adv}$ and invokes $\SET_r(\advst)$ and $\SET_s(\advst)$, thereby the initial state is $\advst_0 = \advst$.

    \item $\adv$ queries $\SNDCO(\ell)$ twice for some length $\ell \geq 1$. The challenger samples and stores two messages $\msgStore[0]$ and $\msgStore[1]$, thereby advancing $\st[s]$ twice, and returns $(\sfHeader_1, \cmsg_1)$ and $(\sfHeader_2, \cmsg_2)$ to $\adv$. As both messages are independent and uniformly distributed, $\msgStore[0] \neq \msgStore[1]$ except with probability $2^{-8\ell}$.

    \item $\adv$ calls $\SET_r(f_\mathsf{adv}(\advst))$; this is permitted as \RCV has not yet been called. As $\advst$ is deterministic, $f_\mathsf{adv}$ exists and is known to the adversary.

    \item $\adv$ calls $\RCV(\sfHeader_2, \cmsg_2)$. Inside \sfdec, the receiver state~$\st[r]$ matches the corresponding sender state~$\st[s]$ used to encrypt $\cmsg_2$, as $\intst = \varepsilon$ and $\advst$ was deliberately set in the previous step. Consequently, $\cmsg_2$ decrypts to $\msgStore[1]$. Since $\idx[r] = 0$, the oracle compares the result against $\msgStore[0]$; as $\msgStore[1] \neq \msgStore[0]$, it sets $\res = \msgStore[1] \neq \bot$.

    \item $\adv$ outputs $\challbitp = 1$ if \RCV returns a value other than $\bot$, and $\challbitp = 0$ otherwise. For $\challbit = 1$ this yields $\challbitp = 1$ whenever $\msgStore[0] \neq \msgStore[1]$, i.e.\ with probability $1 - 2^{-8\ell}$; for $\challbit = 0$ the oracle always returns $\bot$, so $\challbitp = 0$ with certainty. Hence,
    $
        \advantage{\intpcsaco}{\adv} \geq \frac{1}{2} - 2^{-(8\ell+1)}
    $
    is non-negligible.
\end{enumerate}
\end{proof}

\begin{corollary}
    $\sfAlg^\ssh_\mathsf{CCP}$ is not \intpcsaco-secure.
\end{corollary}

\begin{proof}
    By \cref{tab:supportedaead}, the scheme's adversary-chosen state consists solely of the sequence number ($\advst = \sqn$) and the secure state is empty ($\intst = \varepsilon$). The sequence number is deterministic by $f_\mathsf{adv}(\advst) \equiv \advst + 1 \bmod 2^{32}$ and the corollary follows directly from \cref{th:onlyext}.
\end{proof}

\subsection{Insecurity of \glsfmtlong{etm} under \glsfmtshort{co}}
\label{subsec:etmattack}

For \gls{etm} modes, the state is the adversary-chosen sequence number~$\sqn$ together with additional secure state: a chained \gls{iv} (\gls{cbc}), a counter (\gls{ctr}), or a keystream position (Stream). As the secure state is no longer empty, the generic attack of \cref{subsec:genattack} does not apply directly. The same attack nonetheless breaks \intpcsaco security. As shown in \cref{lst:sfSshEtM}, the tag is computed over the sequence number, the packet length, and the ciphertext---that is, over $\sqn \concat \len{\msg'} \concat \cmsg$; the secure state never enters the \gls{mac}. Leaving the ciphertext unchanged and only shifting~\sqn via $\SET_r$ therefore keeps the \gls{mac} valid, so the adversary from \cref{subsec:genattack} applies without change. The secure state is used only by the cipher, affecting only the decrypted plaintext. It is updated by the cipher---using the previous ciphertext block in \gls{cbc}, the number of blocks in \gls{ctr}, or the number of bytes in Stream---so it depends on the earlier traffic, not on~$\sqn$, which only counts packets. The two are therefore independent. If this secure state does not match, the decrypted plaintext is wrong, but the tag still verifies since it covers the ciphertext without the secure state. While such a plaintext may be rejected at the application layer~\cite{USENIX:BauBriSch24}, we exclude this consideration from our model (cf.\ \cref{subsec:restrictions}).

\begin{theorem}
    $\sfAlg^\ssh_{\etmcbc}$, $\sfAlg^\ssh_{\etmctr}$, and $\sfAlg^\ssh_{\etmstr}$ are not \intpcsaco-secure.
\end{theorem}

\begin{proof}
    We reuse the adversary~$\adv$ from the proof of \cref{th:onlyext}, excluding the first step, against the \intpcsaco game of \cref{lst:model} instantiated with one of the \gls{etm} schemes of \cref{lst:sfSshEtM}: $\adv$ queries $\SNDCO(\ell)$ twice, calls $\SET_r(1)$, queries $\RCV(\sfHeader_2, \cmsg_2)$, and decides as before. Its analysis rests on two properties, both of which still hold. The argument is identical for all three modes, which differ only in the additional state.

    \emph{(i) Verification succeeds.} The ciphertext $\cmsg_2$ was produced under $\st[s].\sqn = 1$, so its tag is $\mac(\key_a, 1 \concat \len{\msg_2'} \concat \cmsg_2)$. As $\SET_r(1)$ sets $\st[r].\sqn = 1$ and $\adv$ forwards $(\sfHeader_2, \cmsg_2)$ unchanged, \sfdec recomputes the tag over the identical input, so verification succeeds and returns a plaintext rather than~$\bot$. The additional state does not enter the \gls{mac}, so this is independent of its value.

    \emph{(ii) The accepted plaintext is fresh.} Since $\SET_r$ shifts only $\st[r].\sqn$, the receiver decrypts $\cmsg_2$ under the additional state from \sfsetup rather than the one used to encrypt it, recovering some $\msg^\ast \neq \msgStore[1]$ in general. Decryption of a fixed ciphertext is a bijection, so $\msg^\ast$ is uniform and independent of $\msgStore[0]$, since $\msgStore[1]$ was sampled freshly. With $\idx[r] = 0$, the \RCV oracle thus sets $\res = \msg^\ast \neq \bot$ unless $\msg^\ast = \msgStore[0]$, which occurs with probability~$2^{-8\ell}$.

    Properties (i) and (ii) are what the proof of \cref{th:onlyext} requires, so it carries over and
    $
        \advantage{\intpcsaco}{\adv} \geq \frac{1}{2} - 2^{-(8\ell+1)}
    $
    is non-negligible.
\end{proof}

\subsection{Security of \glsfmtshort{gcm} under \glsfmtshort{cpa}}
\label{subsec:gcmproof}

In \gls{ssh}'s \gls{gcm} mode (\cref{lst:sfAnyXcm}), the sequence number is unused: the nonce is $\aeadNonce = \salt \concat \ictr$ and the associated data is the packet length; \sqn enters neither of them. Consequently, the \SET oracles, which write only \sqn, cannot influence the channel, and \gls{gcm} resists sequence number manipulation. This matches the finding of~\cite{USENIX:BauBriSch24} that \gls{gcm} is unaffected by the Terrapin attack.

Security therefore rests entirely on the invocation counter~\ictr. We prove it under the assumption that \ictr is never reused, is \emph{independent of} \sqn and aligned with the bookkeeping indices up to a fixed offset~$\delta$, that is, $\ictr_s = \idx[s] + \delta$ and $\ictr_r = \idx[r] + \delta$ with $\ictr_r$ advancing only on accepted packets. The independence assumption is essential. If an implementation instead computed the invocation counter from the sequence number---say $\ictr = \ictr_0 + \sqn$---then \ictr would be a function of \sqn, collapsing the state to \sqn. The generic attack of \cref{subsec:genattack} then applies, and the channel would \emph{not} be \intpcsaco-secure.

\begin{theorem}
\label{th:aesgcm}
    Under the \ictr assumption above, $\sfAlg_\gcmmode^\ssh$ is \intpcsacpa-secure, with \[\advantage{\intpcsacpa}{\adv} \leq \epsilon_{\gcmmode}^{\intctxt}.\]
\end{theorem}
\begin{proof}[Sketch]
    \textbf{Game $G_0$} is the original \intpcsacpa game. By assumption the effective state on each side is the invocation counter, $\st[s] = \ictr_s = \idx[s] + \delta$ and $\st[r] = \ictr_r = \idx[r] + \delta$. Since the nonce is $\aeadNonce = \salt \concat \ictr$ and $\sqn$ enters neither nonce nor associated data, the \SET oracles---which write only \sqn---leave \SND/\RCV unaffected. The packet sent at index~$i$ uses nonce $\salt \concat (i + \delta)$; these nonces are pairwise distinct in~$i$.

    \textbf{Game $G_1$} aborts whenever \RCV accepts a tuple $(\sfHeader, \cmsg)$ for which the \gls{gcm} triple $(\aeadNonce, \aeadAD, \cmsg \concat \aeadTag)$ reconstructed by $\sfdec$---with nonce $\aeadNonce = \salt \concat (\idx[r] + \delta)$ from the receiver's state and associated data $\aeadAD = \sfHeader$ from the adversary's input---was never produced by an \SND query. Accepting such a fresh triple constitutes a \gls{gcm} ciphertext forgery, so $|\Pr[G_0] - \Pr[G_1]| \leq \epsilon_\gcmmode^\intctxt$. Hence, from $G_1$ on, $\RCV$ returns~$\bot$ on every query whose triple was not produced by an \SND query.

    \textbf{Game $G_2$} records, for each $\SND$ output, the pair $(\sfHeader, \cmsg)$ and the index~$j$ at which it was created; this is bookkeeping only, so $\Pr[G_2] = \Pr[G_1]$. Consider a query $\RCV(\sfHeader, \cmsg)$. The receiver reconstructs its triple under nonce $\salt \concat (\idx[r] + \delta)$; this is unique because each invocation counter is used only once; by $G_1$ the query is rejected unless $(\sfHeader, \cmsg)$ equals that \SND output; in this case decryption returns the recorded message $\msgStore[\idx[r]]$, and since $\msg = \msgStore[\idx[r]]$, the \RCV oracle sets $\res = \bot$. In both cases, $\RCV$ returns~$\bot$ regardless of the challenge bit, so $\adv$ has advantage~$0$ in $G_2$.
\end{proof}

\subsection{Insecurity of \glsfmtshort{eam}-\glsfmtshort{ctr}/-Stream under \glsfmtshort{kpa}}
\label{subsec:cpaeam1}

For \gls{eam}, the tag is $\mac(\key_a, \sqn \concat \msg')$ (\cref{lst:sfSshEaM}). For \gls{ctr} and Stream variants, where $\cmsg = \msg' \oplus ks$, this breaks \gls{kpa} integrity. An adversary that learns the plaintext~$\msg'$ also learns the keystream $ks = \cmsg \oplus \msg'$, and can therefore re-encrypt any plaintext for which it holds a valid tag under the receiver's keystream; $\SET_r$ then aligns the sequence number to match that tag, and the receiver accepts an out-of-order message. The attack requires \gls{kpa} rather than \gls{co}, since recovering $ks$ needs the plaintext, and it is specific to these modes.

\begin{theorem}
\label{th:eamkp}
    $\sfAlg^\ssh_\eamctr$ and $\sfAlg^\ssh_\eamstr$ are not \intpcsakpa-secure.
\end{theorem}

\begin{proof}
We construct a \gls{ppt} adversary $\adv$ that proceeds as follows:
\begin{enumerate}
    \item $\adv$ queries $\SNDKPA(2\blkAlign - 9)$ twice.\footnote{The term $2\blkAlign - 9$ ensures minimal 4-byte padding even if $\blkAlign = 8$ in the case of \gls{eam}-Stream, considering that there are 5~bytes of length fields.} The oracle samples random $\msg_1, \msg_2$ and returns $(\varepsilon, \cmsg_i \concat \aeadTag_i, \msg_i)$, so $\adv$ learns both plaintext and ciphertext; they differ except with probability $2^{-8(2\blkAlign-9)}$. This yields an encryption input with minimal padding:
    \[
        \msg_i' = (2\blkAlign - 4) \concat \texttt{04} \concat \msg_{i,1}\!\cdots \msg_{i,2\blkAlign-9} \concat pad_{i,1}\!\cdots pad_{i,4},
    \]
    where $\aeadTag_i = \mac(\key_a, (i - 1) \concat \msg_i')$ and $\cmsg_i = \msg_i' \oplus ks_i$, with $ks_i$ the $i$-th packet's keystream. From $\msg_1'$, $\adv$ recovers $ks_1 = \cmsg_1 \oplus \msg_1'$ on all but the four padding bytes and keeps the valid \gls{mac} pair $(1 \concat \msg_2',\, \aeadTag_2)$; $\cmsg_2$ is unused. Now $\msgStore[0] = \msg_1$ and $\msgStore[1] = \msg_2$.

    \item $\adv$ calls $\SET_r(1)$, setting $\st[r].\sqn = 1$; this is admissible, as $\setflag[r] = 0$. Since $\SET_r$ only writes \sqn, the receiver's keystream is unchanged, so the first \RCV still uses $ks_1$.

    \item $\adv$ calls $\RCV(\varepsilon,\, \cmsg \concat \aeadTag_2)$ with $\cmsg = \msg_2' \oplus ks_1$. The oracle decrypts $\cmsg \oplus ks_1 = \msg_2'$ and checks $\mac(\key_a, \sqn \concat \msg_2')$; as $\st[r].\sqn = 1$ this equals $\aeadTag_2$, so $\sfdec$ returns $\msg_2$. With $\idx[r] = 0$ the oracle compares $\msg_2$ against $\msgStore[0] = \msg_1$, and since $\msg_2 \neq \msg_1$ it sets $\res = \msg_2 \neq \bot$.

    \item $\adv$ outputs $\challbitp = 1$ if \RCV returns a value other than $\bot$, and $\challbitp = 0$ otherwise.
\end{enumerate}
$\adv$ knows neither $ks_1$ nor the four padding bytes of $\msg_2'$, so the receiver recovers $\msg_2'$ and $\aeadTag_2$ verifies only with probability $2^{-32}$, giving $\advantage{\intpcsakpa}{\adv} \approx 2^{-33}$, which is non-negligible.
\end{proof}

\subsection{Security of \glsfmtshort{eam}-\glsfmtshort{ctr}/-Stream under \glsfmtshort{co}}
\label{subsec:cpaeam3}

Switching from \SNDKPA to \SNDCO removes the known payload that enabled the attack of \cref{subsec:cpaeam1}, leaving only structural fields known. The adversary can compute the first five bytes of $\msg'$ from the byte length $\len{\msg}$ it selects via \SNDCO: the packet length equals $\len{\msg} + n + 1$ and the padding length is $n$, where $n \geq 4$ with $\len{\msg} + n + 5 \equiv 0 \pmod \blkAlign$ and $n$ minimal. The adversary thus learns the keystream on five bytes of every packet, separated by the remaining $\blkAlign - 5$ random bytes of $\msg$ and $pad$. As $\blkAlign = \Omega(\secpar)$ in our model, this unknown gap grows with $\secpar$.

Both modes' encryption functions are structurally equivalent XOR stream ciphers: to encrypt $\msg'$, they compute a keystream $ks$ using a \gls{prg} and then compute $\cmsg = \msg' \oplus ks$. The \gls{mac} is computed over $\msg'$, so a forgery must reproduce a sent plaintext exactly. For the following proof, we assume that for \gls{ctr} the counter value is never reused, and for Stream the internal state of the keystream generator never repeats. In \gls{ssh}, this is achieved by regularly rotating the session key via a new key exchange.

\begin{theorem}
    \label{th:eamco}
    Let $\sfAlg$ be the \gls{eam} stateful encryption scheme over an XOR stream cipher $\enc$ and a \gls{mac} \mac. Then, for every \gls{ppt} adversary $\adv$, \[\advantage{\intpcsaco}{\adv} \leq \epsilon^{\sufcma}_{\mac} + \epsilon^{\prg}_{\enc} + 2^{-8(B-5)}.\]
\end{theorem}
\begin{proof}[Sketch]
    We argue over a sequence of games; w.l.o.g.\ packets are minimal and span one $\blkAlign$-block, of which the adversary knows the five structural bytes.

    \textbf{Game $G_0$} is the original \intpcsaco game.

    \textbf{Game $G_1$} aborts if \RCV accepts a pair $(\sqn_r \concat \msg^\ast, \aeadTag)$ that $\mac$ never produced during any \SND query, giving $|\Pr[G_0] - \Pr[G_1]| \leq \epsilon^{\sufcma}_{\mac}$.

    \textbf{Game $G_2$} replaces the keystream \gls{prg} with a uniformly random keystream, so $|\Pr[G_1] - \Pr[G_2]| \leq \epsilon^{\mathrm{PRG}}_{\enc}$.

    In $G_2$ the keystream is uniform. To make \RCV return a value other than $\bot$, the adversary must submit $\cmsg$ with $\cmsg \oplus ks_r = \msg_j'$ under the receiver's keystream block $ks_r$. It knows $ks_r$ on at most the five structural bytes (revealed by the packet sent under that state); the remaining $\blkAlign - 5$ bytes are uniform and unseen, so the single forgery attempt that the adversary can make before the channel terminates succeeds with probability at most $2^{-8(\blkAlign-5)}$.
\end{proof}

For \gls{ctr}, the \gls{prg} is a block cipher in counter mode, which is a \gls{prg} whenever the cipher is a \gls{prp}; the term $\epsilon^{\prg}_{\enc}$ reduces to the cipher's \gls{prp} advantage (up to the usual birthday bound) and is negligible, so \gls{eam}-\gls{ctr} is \intpcsaco-secure.

For Stream mode with RC4, the same reduction would apply but it is unclear whether the premise holds: \gls{rc4}'s keystream has positional biases~\cite{USENIX:ABPPS13}, so $\epsilon^{\prg}_{\mathsf{RC4}}$ may be non-negligible.

\subsection{Insecurity of \glsfmtshort{eam}-\glsfmtshort{cbc} under \glsfmtshort{cpa}}
\label{subsec:cpaeam2}

\gls{eam} with \gls{cbc} is not \intpcsacpa-secure: an adaptive chosen-plaintext adversary can mount a probabilistic forgery in the style of the \gls{beast} attack against \gls{tls}.

\begin{theorem}
\label{th:eamcbccpa}
    $\sfAlg^\ssh_\eamcbc$ is not \intpcsacpa-secure.
\end{theorem}

\begin{proof}
Let $\mathit{pad}$ denote the four random padding bytes of message~$2$ below, which are unknown to the adversary, and let $\mathit{pad}^\ast$ be the adversary's guess for them. Write $T = \byte{f2}^{\blkAlign-4}\concat \mathit{pad}^\ast$. The following sequence causes the \SND oracle to emit three messages of two ciphertext blocks each; the first two carry $4$ padding bytes, the third $\blkAlign$ padding bytes:
\begin{align*}
\renewcommand{\arraycolsep}{2pt}
\begin{array}{rcl}
(\varepsilon, cb_0\concat cb_1\concat\aeadTag_0)
  &\gets& \SND(\byte{f0}^{\blkAlign-5}\concat \byte{f1}^{\blkAlign-4}),\\[0.2em]
(\varepsilon, cb_2\concat cb_3\concat\aeadTag_1)
  &\gets& \SND(\byte{f0}^{\blkAlign-5}\concat \byte{f2}^{\blkAlign-4}),\\[0.2em]
(\varepsilon, cb_4\concat cb_5\concat\aeadTag_2)
  &\gets& \SND\bigl((T \oplus cb_3 \oplus cb_0)_{5\dots \blkAlign-1}\bigr),\\[0.2em]
  &&\SET_r(1),\\[0.2em]
  m^\ast
  &\gets&
  \RCV(\varepsilon, cb_0\concat cb_4\concat\aeadTag_1).
\end{array}
\end{align*}
The first two messages have the same length and the same $\byte{f0}^{\blkAlign-5}$ prefix, so their first plaintext block is identical; call it $P_0$. Their second plaintext blocks are $\byte{f1}^{\blkAlign-4}\concat \mathit{pad}'$ and $\byte{f2}^{\blkAlign-4}\concat \mathit{pad}$, respectively. Since \gls{cbc} chains the IV, $cb_0$ is $P_0$ encrypted under the initial IV, and message~$3$ is encrypted under $cb_3$. Its first plaintext block is $P_1 = \mathit{len}\concat\mathit{padlen}\concat m_3$, where $\mathit{len}\concat\mathit{padlen}$ is the $5$-byte header added by \SND and $m_3 = (T \oplus cb_3 \oplus cb_0)_{5\dots \blkAlign-1}$ are the $\blkAlign-5$ payload bytes chosen by the adversary; the second block $cb_5$ is not used.

At the receiver, the IV is still the initial one (nothing has been received), and $\SET_r(1)$ sets its sequence number to $1$, the index under which message~$2$ and its tag $\aeadTag_1$ were sent. The receiver decrypts $cb_0\concat cb_4$ block by block. The first block $cb_0$ decrypts back to $P_0$, which is also the first block of message~$2$. The second block decrypts to $\dec(cb_4)\oplus cb_0 = P_1 \oplus cb_3 \oplus cb_0$, where $\dec$ is the block-cipher inverse. By the choice of $m_3$,
\begin{align*}
  \dec(cb_4)\oplus cb_0 =
  \bigl(\mathit{len}\concat\mathit{padlen} \oplus (cb_3 \oplus cb_0)_{0\dots 4}\bigr)
  \concat T_{5\ldots \blkAlign-1}.
\end{align*}
The tag $\aeadTag_1$ verifies exactly when this block equals message~$2$'s second block $\byte{f2}^{\blkAlign-4}\concat \mathit{pad}$. Bytes $5,\dots,\blkAlign-1$ agree iff the guess $\mathit{pad}^\ast$ matches $\mathit{pad}$ on the positions it covers; bytes $0,\dots,4$ agree iff the fixed header $\mathit{len}\concat\mathit{padlen}$ equals $(\byte{f2}^{\blkAlign-4}\concat\mathit{pad} \oplus cb_3 \oplus cb_0)_{0\dots 4}$, an event over the block-cipher outputs $cb_0$ and $cb_3$. Altogether, $\min(9, \blkAlign)$ bytes are outside the adversary's control, so the forgery succeeds with probability $2^{-72}$ for $\blkAlign \geq 9$ and $2^{-64}$ for $\blkAlign = 8$, where byte~$4$ of the target block is $\mathit{pad}_0$ rather than $\byte{f2}$. On success, \RCV returns message~$2$'s plaintext, which differs from message~$1$ stored in $\msgStore[0]$, so the receiver accepts a different message at position~$0$. The adversary outputs $\challbitp = 1$ if \RCV\ returns a value other than $\bot$ and $\challbitp = 0$ otherwise; as a non-$\bot$ answer occurs only when $\challbit = 1$ and the forgery succeeds, $\advantage{\intpcsacpa}{\adv} \geq 2^{-73}$ for every $\blkAlign \geq 8$.
\end{proof}

The probability $2^{-72}$ does not yield a practically exploitable attack. Each attempt transfers $8$ blocks, so achieving probable success requires $8\cdot 2^{72}= 2^{75}$ blocks---more than the $\approx 2^{64}$ birthday bound of a $128$-bit cipher, so in that case a collision attack is cheaper. However, our attack's cost is fixed in the block length $\blklen$, while a collision attack costs $\approx 2^{\sfrac{\blklen}{2}}$: for $\blklen \geq 150$ the comparison flips. \gls{ssh} rekeying has the same effect. With a rekey interval of $l$ blocks, the per-key collision probability is only $\sfrac{l(l-1)}{2^{\blklen+1}}$, requiring $\approx \sfrac{2^{\blklen+1}}{l}$ blocks: about $2^{97}$ for $l = 2^{32}$, the interval recommended by~\cite{rfc4344}.

\subsection{Security of \glsfmtshort{eam}-\glsfmtshort{cbc} under \glsfmtshort{kpa}}
\label{subsec:cpaeam4}

For the following proof, we assume the number of \SND invocations is bounded such that birthday collisions on the block cipher are negligible. In \gls{ssh}, this is achieved by timely rekeying of the session key.

\begin{theorem}
\label{th:eamcbckpa}
    $\sfAlg^\ssh_\eamcbc$ is \intpcsakpa-secure. For every \gls{ppt} adversary, the advantage is
    \[
        \advantage{\intpcsakpa}{\adv} \leq \epsilon^{\sufcma}_{\mac} + \epsilon^{\sprp}_{\enc} + 2^{-8B}.
    \]
\end{theorem}

\begin{proof}[Sketch]
    We argue over a sequence of games; w.l.o.g.\ packets are minimal and span one $\blkAlign$-block.

    \textbf{Game $G_0$} is the original \intpcsakpa game.

    \textbf{Game $G_1$} aborts if \RCV accepts a pair $(\sqn_r \concat \msg', \aeadTag)$ that $\mac$ never produced during an \SND query, giving $|\Pr[G_0] - \Pr[G_1]| \leq \epsilon^{\sufcma}_{\mac}$.

    \textbf{Game $G_2$} replaces the block cipher encryption $\enc$ with a random permutation $\pi$, so $|\Pr[G_1] - \Pr[G_2]| \leq \epsilon^{\sprp}_{\enc}$. In $G_2$ the permutation is uniform. \gls{iv} chaining couples the receiver state---an accepted \RCV advances $\st[r].\sqn$, $\st[r].\IV$, and $\idx[r]$ in lockstep---and the only oracle that moves $\st[r].\sqn$ on its own is $\SET_r$,  which is admissible only before the first \RCV, when $\st[r].\IV$ is still the secret initial $\IV[\mathsf{init}]$. If the adversary advances the receiver by delivering genuine packets---learning $\st[r].\IV$ via chaining---then by $G_1$ any accepted \RCV must decrypt to the unique $\msg_k'$ tagged under $\st[r].\sqn$. As CBC decryption under a fixed \gls{iv} is a bijection, the submitted ciphertext must be the matching \SND output, so stateful decryption returns $\msgStore[\idx[r]]$ and gives $\res=\bot$; a known \gls{iv} is useless. To reach a fresh target, the adversary must therefore use $\SET_r$, possible only before the first \RCV, where $\st[r].\IV=\IV[\mathsf{init}]$; the fail-stop ($\rcvfailflag$) then permits a single attempt. As only the first packet was encrypted under $\IV[\mathsf{init}]$ and a fresh target has $k\neq0$, the ciphertext $\pi(\msg_k'\oplus \IV[\mathsf{init}])$ is a fresh output of $\pi$, matching with probability $2^{-8\blkAlign}$.
\end{proof}

\section{Summary and Outlook}
\label{sec:outlook}
In our extended theoretical model, we have shown that if a vulnerability corresponding to a \SET oracle exists in an implementation, four out of eight stateful encryption schemes can be broken, while the security of one additional scheme depends on the security of \gls{rc4}. In \gls{ssh} without the strict KEX countermeasure, the protocol exhibits a real-world \SET oracle, as shown in~\cite{USENIX:BauBriSch24}. When considering \glsxtrfullpl{kpa} and \glsxtrfullpl{cpa} adversaries, all modes are insecure except \gls{gcm} and, for \gls{kpa} only, \gls{eam}-\gls{cbc}. These results do not consider layers above the \glsxtrlong{bpp}, which on the one hand implement additional validity checks but on the other hand have known-plaintext components and might even allow limited chosen plaintext. We leave a study of the \gls{kpa}/\gls{cpa} potentials of these layers as future work.

As for insecure modes, we note that not all attacks are equally effective. In particular, the \gls{kpa}-insecurity of \gls{eam}-\gls{ctr} and \gls{eam}-Stream, as well as the \gls{cpa}-insecurity of \gls{eam}-\gls{cbc}, depend on the minimum padding length, which is 4 bytes in \gls{ssh}, independent of the block length. A modification to the \gls{ssh} specification or its implementations to increase the minimum padding length could strengthen the security of these modes in our model. Finally, we cannot prove the security of RC4-based stream cipher modes in \gls{ssh} mainly due to defects in \gls{rc4} itself. On the other hand, we cannot provide an attack against \gls{rc4}, leaving the question of its security in \gls{ssh} channels open.

There are many interesting constructions of secure channels that deserve a similar treatment as \gls{ssh}. First, there are \emph{standardized} channels like \gls{tls}, \gls{ipsec} \gls{esp}/\gls{ah}~\cite{RFC4303,RFC4302,EC:PatYau06,SP:DegPat07,CCS:DegPat10}, \gls{dtls}~\cite{RFC6347,RFC9147,SP:AlFPat13,USENIX:EMMSS23}, \gls{srtp}, and, perhaps most interestingly, QUIC~\cite{RFC9000,EPRINT:FisGunJan20}. Second, there are non-\gls{ietf} protocols like tcpcrypt, MinimaLT, CurveCP, WireGuard, and OpenVPN (see~\cite{RFC8922} for references). Some of these protocols are widely deployed, but reliable empirical data is lacking. To better understand the security of these schemes, a thorough analysis of the stateful encryption layer (and the \gls{ake} protocol used) would be desirable.

\bibliographystyle{splncs04}
\bibliography{bib/abbrev2,bib/crypto,bib/rfc_official,bib/paper}

\clearpage
\appendix

\section{Ethical Considerations}
\label{sec:ethics}
\subsection{Stakeholder Analysis}

The results of this work are immediately relevant to protocol designers and standards bodies, such as the \gls{ietf}, who can use them to improve their understanding and design of internet standards that rely on \gls{aead} and stateful encryption. In addition, they are relevant to security researchers, who may apply comparable methodologies to other protocols. Our findings can help \gls{ssh} library implementers prioritize the selection of stateful encryption modes, enhance their documentation, and warn users about possible dangers. Our results may further aid in finding appropriate mitigations for any weaknesses similar to the Terrapin attack that may be found in the future. Indirect benefits are also obtained by end users of products that use these protocols.

We acknowledge that adversaries can also benefit from improved understanding of protocol weaknesses; for example, they might be able to identify new attacks more efficiently by misusing the information in this work. In our assessment, the defensive benefits significantly outweigh any potential advantage given to adversaries.

\subsection{Impact}

In light of ethical standards, we considered how this work would benefit the aforementioned stakeholders, particularly in terms of enhancing protocol security and addressing state-based vulnerabilities in our adversarial model. We also considered the possible advantages of informing the design of future protocols and creating strong defenses against the Terrapin attack and other vulnerabilities of a similar nature. Our study did not involve human subjects, personal information, live system interaction, or other consent- or privacy-related concerns.

Our findings apply to a wide range of protocols and do not disproportionately impact any specific user group. The decision to focus on \gls{ssh} rather than similar protocols such as \gls{tls}, \gls{ipsec}, \gls{dtls}, \gls{srtp}, and QUIC reflects the Terrapin attack's focus on \gls{ssh}. We believe that a thorough examination of a single protocol benefits stakeholders more than a broader but more superficial study.

This work adheres to the established norms of responsible security research and does not present any new practical attacks beyond those in the original Terrapin publication~\cite{USENIX:BauBriSch24}, out of respect for law and the public interest. Therefore, we believe that conducting and publishing this research is unlikely to cause direct harm to stakeholders. No responsible disclosure was required, given that the Terrapin attack had already been disclosed in~\cite{USENIX:BauBriSch24}.

\subsection{Mitigations}

Our analysis is abstract and model-based and does not directly lead to exploitable attacks in real implementations. The focus of this work is on classification and understanding of stateful encryption and \gls{aead} modes. These results can be used to implement robust mitigations against the Terrapin attack and potential future attacks similar to it, as well as systematically harden protocol designs against such vulnerabilities. In this way, our work encourages safer cryptographic design and deployment.

\subsection{Decision to Publish}

We believe that the stakeholders protecting against attacks would suffer more from withholding our findings than adversaries would. Adversaries currently possess the Terrapin attack method, and their capabilities would not be significantly enhanced by withholding a formal analysis. Defenders, on the other hand, need to systematically find and fix pertinent flaws in protocols and implementations. For these reasons, we conclude that publication of this work is ethically justified.

\section{Generative AI Usage}
\label{sec:ai}
For this work, we used generative AI, namely ChatGPT, Claude, and Google Gemini, to proofread the paper and identify possible gaps. All issues flagged by the AI were manually verified and, where necessary, corrected. Furthermore, we used these tools alongside LanguageTool to improve grammar and spelling, and for light style polishing, that is, by rewriting individual sentences for the final submission.

\clearpage
\section{Formal Model for Indistinguishability of Ciphertexts}

\Cref{lst:model2} is an extension of our Terrapin security model. It replaces the \gls{cpa} oracle $\SND(\msg)$ with a left-or-right oracle $\SND(\msg_0, \msg_1)$. Note that this security model is strictly stronger than the Terrapin model, since we can use the left-right oracle as a \gls{cpa} oracle by calling $\SND(\msg, \msg)$.

\begin{figure*}
    \input{img/cryptocode/ind-cca-mbe}
    \caption{Our model for ciphertext indistinguishability under adaptive chosen-ciphertext attacks. The pseudocode defines a single-instance security game played between an adversary $\adv$ and a \indpcsaccaii challenger, where the adversary is provided access to the \SND~and \RCV~oracles as well as the \SET oracles from \cref{lst:model}. The array $\msgStore$ records which message is allowed at which index in the channel.}
    \label{lst:model2}
\end{figure*}

\end{document}